\documentclass[cmp,envcountsect,smallextended]{svjour}

\usepackage{graphics}

\usepackage{amsfonts,amssymb,amsmath}
\usepackage{mathrsfs}
\usepackage{graphicx}
\usepackage{color}
\usepackage[normalem]{ulem}
\usepackage{indentfirst}
\usepackage{amsrefs}
\usepackage{enumerate}

\numberwithin{equation}{section}

\DeclareMathOperator{\sgn}{sgn}

\newcommand{\ud}{\,\mathrm{d}}
\newcommand{\abs}[1]{\left\lvert#1\right\rvert}
\newcommand{\norm}[1]{\lVert#1\rVert}
\newcommand{\average}[1]{\langle#1\rangle}

\newcommand{\RR}{\mathbb{R}}

\newcommand{\wt}[1]{\widetilde{#1}}

\newcommand{\mc}[1]{\mathcal{#1}}

\newcommand{\eps}{\epsilon}

\newcommand{\wconv}{\rightharpoonup}

\newcommand{\oH}{\mathring{H}}

\newcommand{\beq}{\begin{equation}}
\newcommand{\eeq}{\end{equation}}

\begin{document}
\title{Orbital-free density functional theory of out-of-plane charge
  screening in graphene}
\author{Jianfeng Lu\inst{1}
\and
Vitaly Moroz\inst{2}
\and
Cyrill B.~Muratov\inst{3}
}                     
\institute{Departments of Mathematics, Physics, and Chemistry, Duke
  University, Box 90320, Durham, NC 27708, USA. \\
  \email{jianfeng@math.duke.edu} \and Swansea University, Department
  of Mathematics, Singleton Park, Swansea SA2~8PP, Wales, United
  Kingdom. \\
  \email{v.moroz@swansea.ac.uk} \and Department of Mathematical
  Sciences, New Jersey Institute of Technology,
  University Heights, Newark, NJ~07102, USA. \\
  \email{muratov@njit.edu} }
\date{\today}
%
%
\maketitle
\begin{abstract}
  We propose a density functional theory of Thomas-Fermi-Dirac-von
  Weizs\"acker type to describe the response of a single layer of
  graphene resting on a dielectric substrate to a point charge or a
  collection of charges some distance away from the layer. We
  formulate a variational setting in which the proposed energy
  functional admits minimizers, both in the case of free graphene
  layers and under back-gating. We further provide conditions under
  which those minimizers are unique and correspond to configurations
  consisting of inhomogeneous density profiles of charge carrier of
  only one type. The associated Euler-Lagrange equation for the charge
  density is also obtained, and uniqueness, regularity and decay of
  the minimizers are proved under general conditions. In addition, a
  bifurcation from zero to non-zero response at a finite threshold
  value of the external charge is proved.
\end{abstract}

%
\section{Introduction}
\label{sec:introduction}

Graphene is a two-dimensional monolayer of carbon atoms arranged into
a perfect honeycomb lattice \cite{novoselov04}. It has received a huge
amount of attention in recent years, both as a very promising material
for nanotechnology applications and as a model system with pronounced
quantum mechanical properties (for reviews, see
\cite{geim07,castroneto09,abergel10}). The current interest in
graphene stems from its very unusual electronic properties closely
related to the symmetry and the two-dimensional character of the
underlying crystalline lattice, into which the carbon atoms arrange
themselves. A free-standing graphene layer acts as a semi-metal, in
which the low energy charge carrying quasiparticles (electrons and
holes) behave to a first approximation as massless fermions obeying a
two-dimensional relativistic Dirac equation
\cite{wallace47,fefferman12,fefferman12a}. Hence, their kinetic energy
is proportional to their quasi-momentum:
\begin{eqnarray}
  \label{eq:vF}
  \varepsilon_\mathbf{k} = \pm \hbar v_F |\mathbf k|,
\end{eqnarray}
where $v_F \simeq 1 \times 10^8$ cm/s is the Fermi velocity, $\mathbf
k$ is the wave vector and ``$\pm$'' stands for electrons and holes,
respectively. This equation is valid for $|\mathbf k| \ll a_0^{-1}$,
where $a_0 \simeq 1.42 \, \mathring{A}$ is the nearest-neighbor
distance between the carbon atoms in the graphene lattice (without
taking into account the effect of the velocity renormalization
\cite{gonzalez94,KotovUchoa:12,martin08,reed10,sodemann12,yu13}).

In contrast to the fermions with non-zero effective mass in the usual
metals or semiconductors, in graphene the effect of interparticle
Coulomb repulsion does not decrease with increasing carrier density
\cite{KotovUchoa:12}. This can already be seen from simple dimensional
considerations: according to \eqref{eq:vF} a single particle whose
wave function is localized into a wave packet of radius $\sim r$ would
have kinetic energy $E_{kin} \sim \hbar v_F / r$, while the energy of
Coulomb repulsion per particle (in CGS units) is $E_{Coulomb} \sim e^2
/ (\epsilon_d r)$, where $e > 0$ is the elementary charge and
$\epsilon_d$ is the effective dielectric constant in the presence of a
substrate. Thus their ratio $\alpha = e^2 / (\epsilon_d \hbar v_F)$,
which characterizes the relative strength of the Coulombic
interaction, is a constant independent of $r$, and, furthermore, we
have $\alpha \simeq 2.2$ for $\epsilon_d = 1$, indicating the
non-perturbative role of the Coulombic interaction in the absence of a
strong dielectric background.

The scaling argument above can also be applied to an electron obeying
\eqref{eq:vF} in an attractive potential of a positively charged
ion. When the valence $Z$ of the ion increases, the potential energy
$E_{pot} \sim -Z e^2 / r$ of the attractive interaction between the
electron and the ion always overcomes the kinetic energy. At the
single particle level this effect results in non-existence of single
particle ground states for the relativistic Dirac-Kepler problem
\cite{shytov07}, which is somewhat similar to the phenomenon of
relativistic atomic collapse \cite{lieb88}. In a more realistic
multiparticle setting the situation is more complicated due to
strongly correlated many-body effects involving both the electrons and
holes. In fact, exactly how the carriers in graphene screen a charged
impurity is a subject of an ongoing debate, with qualitatively
different predictions for the behavior of the screening charge density
and the total electrostatic potential coming from different theories.

Early studies of screening of the electric field from point charges in
graphene go back to the work of DiVincenzo and Mele, who used a
self-consistent Hartree-type model to analyze the electron response to
interlayer charges in intercalated graphite compounds
\cite{divincenzo84}. They found a surprising result that the screening
electron density decays as $1/r^2$ (to within an undetermined
logarithmic factor), indicating that the screening charge is
considerably spread out laterally within the graphene layer. They also
made a similar conclusion from the analysis of the Thomas-Fermi
equations for massless relativistic fermions and contrasted it with
the $1/r^3$ behavior expected from the image charge on an
equipotential plane in the case of perfect screening.  In sharp
contrast, Shung performed an analysis of the dielectric susceptibility
of intercalated graphite compounds using linear response theory
\cite{shung86}. His calculation implies that in the absence of doping
only partial screening of an impurity should occur and that the
electron system should behave effectively as a dielectric medium due
to the excitation of virtual electron-hole pairs, which has an effect
of renormalizing the value of $\epsilon_d$ (see also
\cite{gonzalez94,Ando:06,HwangDasSarma:07,KotovUchoa:12} for further
discussions). He also commented that the nonlinear effects are of
major importance in the screening, which explains the different
results he had for linear response comparing with the Thomas-Fermi
result in \cite{divincenzo84}.

More recently, Katsnelson computed the asymptotic behavior of the
screening charge density for a charged impurity within the
Thomas-Fermi theory of massless relativistic fermions with a lattice
cutoff at short scales \cite{Katsnelson:06}. He found that the
screening charge density should behave as $1/(r \ln r)^2$ far from the
impurity, refining earlier results of \cite{divincenzo84} and
demonstrating the importance of nonlinear screening effects in
graphene. Fogler, Novikov and Shklovskii further considered the effect
of an out-of-plane hypercritical charge $Z \gg 1$ on the electron
system in a graphene layer and argued for perfect screening ($1/r^3$
behavior of the screening charge density and constant electrostatic
potential in the layer) \cite{fogler07}. They also argued for a
crossover between perfect screening in the near field tail,
Thomas-Fermi screening ($1/(r \ln r)^2$ behavior of the screening
charge density and $1/(r \ln r)$ decay of the electrostatic potential
in the layer) in the far field tail, and partial screening (dielectric
response with no screening charge and $1/r$ decay of the electrostatic
potential) in the very far tail for certain ranges of $Z$ and
$\alpha$. We also note that a recent result indicates that in the
Hartree-Fock approximation the relative dielectric constant of
graphene is, somewhat surprisingly, equal to unity in the Hartree-Fock
theory, implying that the total induced charge from a charged impurity
in graphene is zero (no partial screening or effectively very weak
screening due to the slow decay) \cite{HainzlLewinSparber:2012}.

The differing conclusions of the above works indicate a very delicate
nature of the problem of screening in graphene (see also the
discussion in \cite{KotovUchoa:12} and further references
therein). One reason is the precise tuning of the kinetic energy, the
Coulombic attraction of electrons to the impurity and the Coulombic
repulsion between electrons, which is already evident from the scaling
argument presented earlier. Another reason is that the studies
mentioned above do not account for the correlation effects. While it
is believed that exchange does not play a significant effect in
graphene, correlations between electrons and holes due to their
Coulombic attraction (excitonic effects) may have an effect on the
nature of the response beyond random phase approximation
\cite{KotovUchoa:12,abergel09,abergel10,martin08,reed10,
  sodemann12,wang12,yu13}. Finally, the third reason is that in view
of the crucial role played by nonlinear and nonlocal effects for
charge carrier behavior in graphene the analysis of the problem, both
mathematical and numerical, becomes rather non-trivial.

Our approach to the problem of screening of point charges by a
graphene layer is via introducing a Thomas-Fermi-Dirac-von
Weizs\"acker (TFDW) type energy for massless relativistic fermions and
studying the associated variational problem. The considered energy
functional is a variant of an orbital-free density functional theory
(for a recent Kohn-Sham-type density functional theory see
\cite{polini08}) that models the exchange and correlation effects by
renormalizing the corresponding coefficients of the Thomas-Fermi
theory for the system of non-interacting massless relativistic
fermions and introducing a non-local analog of the von Weizs\"acker
term in the usual TFDW model of a non-relativistic electron gas
\cite{Lieb:81,lebris05}. For simplicity, we begin by treating the
problem of the influence of a single point charge $+Z e$ located at
distance $d \gg a_0$ away from the graphene layer on the electrons in
the layer. It may either correspond to the effect of a charge placed
on a gate separated from the graphene layer by a layer of insulator in
the context of graphene-based nanodevices, or it may correspond to an
imbedded charged impurity or a cluster of impurities within the
dielectric substrate.  After a suitable rescaling, the TFDW energy for
graphene at the neutrality point in the presence of an impurity takes
the following form:
\begin{multline}\label{E0}
  E_0(\rho) = a \int_{\RR^2} \left| \nabla^{\frac12} \bigl(
    \sqrt{\abs{\rho(x)} } \sgn(\rho(x))\bigr) \right|^2 \ud^2 x +
  \frac{2}{3}\int_{\RR^2} \abs{\rho(x)}^{3/2} \ud^2 x  \\
  - \int_{\RR^2} \frac{\rho(x)}{\bigl(1 + \abs{x}^2\bigr)^{1/2}} \ud^2 x +
  \frac{b}{2} \iint_{\RR^2 \times \RR^2} \frac{\rho(x) \rho(y)}{\abs{x
      - y}} \ud^2 x \ud^2 y.
\end{multline}
Here $\rho(x)$ is the signed particle density, with $\rho > 0$
corresponding to electrons and $\rho < 0$ corresponding to holes, and
$a \geq 0$ and $b \geq 0$ are two dimensionless parameters
characterizing the model. Note that in the case of $a = 0$ we recover
the usual Thomas-Fermi model for graphene. The case of $b = 0$ would
correspond to a model system of non-interacting massless relativistic
fermions in an external potential. The meaning of each term in
\eqref{E0} and the relation to the original physical parameters is
explained in Sec. \ref{sec:model}. Let us point out the unusual
non-local nature of both the first and the last terms in
\eqref{E0}. The first term involves the homogeneous $H^{1/2}(\RR^2)$
norm squared $\int_{\RR^2} |\nabla^{\frac12} u|^2 \, \ud^2 x$ of $u =
\rho / |\rho|^{1/2}$, while the last term involves the homogeneous
$H^{-1/2}(\RR^2)$ norm squared of $\rho$. This is in contrast to the
conventional TFDW models of massive non-relativistic fermions, in
which the first term involves the homogeneous $H^1$ norm and the last
term involves the homogeneous $H^{-1}$ norm, respectively. The
difference in the first term has to do with the relativistic character
of the dispersion relation for quasiparticles in graphene at low
energies given by \eqref{eq:vF}, while the difference in the last term
reflects the three-dimensional character of Coulomb interaction and
the two-dimensional character of the charge density. We point out that
a von Weizs\"acker-type term similar to the first term in \eqref{E0}
appeared in the studies of stability of relativistic matter (see,
e.g., \cite{lieb96} and references therein).  We also note that our
model is different from the ultrarelativistic
Thomas-Fermi-von~Weizs\"acker model studied in \cite{EngelDreizler:87,
  EngelDreizler:88, BenguriaLossSiedentop:2008}, where a local
gradient term in the kinetic energy for massless relativistic fermions
in three space dimensions was used. An analogous term for graphene
would have been $\int_{\RR^2} \bigl| \nabla \abs{\rho}^{1/4} \bigr|^2
\ud^2 x$ (see Sec. 2 for the explanation of our choice of the
non-local term).

The model above is easily generalized to include a collection of point
charges or a localized distribution of charges some distance away from
the graphene layer. If
\begin{align}
  \label{eq:V}
  V(x) = - \int_{\RR^3} {\ud \mu(y, z) \over (|1 + z|^2 + |x -
    y|^2)^{1/2}},
\end{align}
where $\mu(y,z)$ is a finite signed Radon measure with compact support
located at $z \geq 0$ in $\RR^3$, e.g., $\mu(y, z) = \sum_{i=1}^N c_i
\delta(y - y_i) \delta(z - z_i)$ with $c_i \in \RR$, $y_i \in \RR^2$
and $z_i \geq 0$ for all $i = 1, \ldots, N$ ($c_i > 0$ would
correspond to positive external charges), then the generalization of
the energy in \eqref{E0} reads
\begin{equation}
  \label{eq:E}
  \begin{aligned}
    E(\rho) & = a \int_{\RR^2} \left| \nabla^{\frac12} \bigl(
      \sqrt{\abs{\rho(x)} } \sgn(\rho(x))\bigr) \right|^2 \ud^2 x \\
    & + \frac{2}{3}\int_{\RR^2} (\abs{\rho(x)}^{3/2} -
    |\bar\rho|^{3/2} ) \ud^2 x - |\bar\rho|^{1/2} \sgn(\bar\rho)
    \int_{\RR^2} ( \rho(x) - \bar \rho) \ud^2 x \\
    & + \int_{\RR^2} V(x) (\rho(x) - \bar \rho) \ud^2 x + \frac{b}{2}
    \iint_{\RR^2 \times \RR^2} \frac{(\rho(x) - \bar \rho) (\rho(y) -
      \bar \rho) }{\abs{x - y}} \ud^2 x \ud^2 y.
  \end{aligned}
\end{equation}
Here we also included the possibility of a net background charge
density $\bar\rho \in \RR$, which can be easily achieved in graphene
via back-gating, and subtracted the divergent contributions of the
background charge density to the energy.

In this paper we establish basic existence results for minimizers of
the energy, which is a slightly generalized version of the one in
\eqref{eq:E}, under some general assumptions on the potential $V$,
which include, in particular, potentials of the form given by
\eqref{eq:V}. We begin by developing a variational framework for the
problem and proving a general existence result among admissible
$\rho$ which may possibly change sign, see Theorem
\ref{thm:Egeneral}. We also establish basic regularity and uniform
decay properties of these minimizers, as well as the Euler-Lagrange
equation solved by the minimizing profile.

We shall emphasize that sign-changing profiles with finite energy
include, in particular, the profiles for which the Coulomb energy term
does not admit an integral representation and shall be understood in
the distributional sense, even if the profile is a continuous function
(see Example \ref{example-nonint}). Mathematically, this makes the
analysis of the problem particularly challenging. It is an interesting
open question whether it is possible for a sign--changing minimizer to
have a Coulomb energy which does not have an integral representation.

We then turn our attention to minimizers among non-negative
 $\rho$. Here we prove in Theorem
\ref{thm:Egenplus} the existence of a unique minimizer in the
considered class in the case of strictly positive background charge
density $\bar\rho$. Importantly, using a version of a strong maximum
principle for the fractional Laplacian, we also show that these
minimizers are strictly positive and, as a consequence, also satisfy
the associated Euler-Lagrange equation. In the next theorem, Theorem
\ref{thm:positive}, we give a sufficient condition that guarantees
that the global minimizer among all admissible $\rho$, including
those that change sign, is given by the unique positive minimizer
obtained in the preceding theorem.

The remaining two theorems are devoted to the case of zero background
charge density. In Theorem \ref{thm:EA0p} we give an existence result
for non-negative minimizers, alongside with strict positivity and
uniqueness.  In Theorem \ref{thm:E0A0}, using a suitable version of
fractional Hardy inequality, we establish a bifurcation result for a
particular problem in which the background potential is given by the
electrostatic potential of a point charge some distance away from the
graphene layer. We also illustrate the conclusion of Theorem
\ref{thm:E0A0} with a numerical example.

Our paper is organized as follows. In Sec. \ref{sec:model}, we discuss
the derivation and justification of different terms in the energy and
connect our model with the physics literature. In
Sec. \ref{sec:statement-results}, we state our main results. In
Sec. \ref{sec:prelim}, we introduce various notations and auxiliary
lemmas that are used throughout the paper. In Sec. \ref{sec:var}, we
formulate the precise variational setting for the minimization
problem. Finally, in Sec. \ref{sec:th31} we prove Theorems
\ref{thm:Egeneral} and \ref{thm:positive}, and in Sec. \ref{sec:th32}
we prove Theorems \ref{thm:Egenplus}, \ref{thm:EA0p} and
\ref{thm:E0A0}.

\section{Model}
\label{sec:model}

Our starting point is the following (dimensional) energy for the
graphene layer in the presence of a single positively charged
impurity:
\begin{multline}
  \label{EE}
  E_0(\rho) = C_{\text{W}}\int_{\RR^2} \left| \nabla^{\frac12}
    \bigl(\sqrt{\abs{\rho(x)}} \sgn(\rho(x))\bigr) \right|^2 \ud^2 x +
  \frac23 C_{\text{TFD}} \int_{\RR^2} \abs{\rho(x)}^{3/2}
  \ud^2 x  \\
  - \frac{Z e^2}{\eps_d} \int_{\RR^2} \frac{\rho(x)}{(d^2 +
    |x|^2)^{1/2}} \, \ud^2 x + \frac{e^2}{2\eps_d} \iint_{\RR^2
    \times \RR^2} \frac{\rho(x) \rho(y)}{\abs{x - y}} \ud^2 x \ud^2 y,
\end{multline}
which is a functional defined on a signed particle density $\rho(x)$
in a flat graphene layer of infinite extent, with the convention that
$\rho > 0$ corresponds to the electron-rich region and $\rho < 0$
corresponds to the hole-rich region (for definiteness, in this
  section we assume $\rho \in C^\infty_c(\mathbb R^2)$). The terms in
\eqref{EE} are, in order: the von Weizs\"acker-type term that
penalizes spatial variations of $\rho$, the Thomas-Fermi-Dirac term
containing both the contribution from the kinetic energy of the
particles and the Dirac-type contribution from exchange and
correlations, the interaction term between the particles and the
external out-of-plane point charge $+Z e$, and the Coulomb self-energy
in the presence of a substrate providing an effective dielectric
constant $\eps_d$.

The energy functional in \eqref{EE} should be viewed as a
semi-empirical model in which the constants $C_\text{W}$,
$C_\text{TFD}$ and $\eps_d$ are to be fitted to the experimental data
for a particular setup. It is easy to see that for an {\em ideal}
uniform gas of non-interacting massless relativistic fermions the
kinetic energy contribution per unit area is given by $\frac23
C_\text{TF}^0 |\rho|^{3/2}$, where $C_{\text{TF}}^0 = \hbar v_F
\sqrt{\pi}$ and the 4-fold quasiparticle degeneracy was taken into
account (see for example \cite{Katsnelson:06, BreyFertig:09,
  DasSarma_RMP, ZhangFogler:08}).\footnote{Note that in
  \cite{BreyFertig:09, DasSarma_RMP} and some other papers in the
  physics literature, a factor of $\sgn(\rho)$ was mistakenly added to
  the integrand of the Thomas-Fermi term. The resulting energy
  functional is then not bounded from below and is inconsistent with
  the Thomas-Fermi equation.}  We note, however, that in real graphene
the Coulombic interaction noticeably renormalizes the Fermi velocity
\cite{gonzalez94,KotovUchoa:12,martin08,reed10,sodemann12,yu13}. In
practice the value of $C_\text{TF}^0$ based on the experimentally
observed value of $v_F$ (at the experimental length scale) includes
the many-body effects due to Coulombic interparticle
forces. Similarly, for $\alpha \ll 1$ the leading order exchange
and correlation contributions per unit area of the ideal uniform gas
of massless relativistic fermions are given by $C_\text{D}^0
|\rho|^{3/2}$, where $C_\text{D}^0 = (c_1 \alpha - c_2 \alpha^2)
C_\text{TF}^0$ and both $c_1$ and $c_2$ weakly (logarithmically)
depend on the ratio of the experimental length scale to $a_0$
\cite{barlas07,sodemann12}. Therefore, in the local approximation the
combined contribution of the kinetic energy and the exchange term
would have, to the leading order in $\alpha$, the form of the second
term in \eqref{EE} with some constant $C_\text{TFD}^0 > 0$. This
conclusion is also confirmed by recent experimental measurements of
inverse quantum compressibility in graphene
\cite{martin08,yu13}. Using the renormalized rather than bare Fermi
velocity may then eliminate the need to consider the additional
exchange and correlation terms, at least on the local level. We also
note that in contrast to the usual TFDW models of massive
non-relativistic fermions \cite{Lieb:81,lebris05}, in graphene the
local approximation to the exchange energy does {\em not} produce a
non-convex contribution to the energy.

We now explain the origin of the first term in \eqref{EE}. Recall that
in the usual TFDW model of massive non-relativistic fermions the
analogous von Weizs\"acker term takes the form $C_\text{W} \int \left|
  \nabla \sqrt{\rho} \, \right|^2 \ud^3 x$, with the constant
$C_\text{W} \sim \hbar^2 / m^*$, where $m^*$ is the effective mass
(recall that for a single parabolic band one has $\rho \geq 0$)
\cite{Lieb:81,lebris05}. The basic rationale for the introduction of
such a term is to penalize spatial variations of $\rho$, favoring
spatially homogeneous ground state density for the system of
non-interacting particles (see also the discussion in
\cite{lieb96}). The choice of the specific form of the integrand is
determined by the following three requirements:
\begin{enumerate}
\item[1)] The energy must scale linearly with $\rho$. \smallskip

\item[2)] The energy must be the square of a homogeneous Sobolev norm
  of $\rho g(|\rho|)$, for some positive scale-free function
  $g$. \smallskip

\item[3)] The energy must scale as the Thomas-Fermi term under
  rescalings of $x$ and $\rho$ that preserve the total number of
  particles.
\end{enumerate}
The first requirement above reflects the extensive nature of the
contributions of individual particles. The second requirement reflects
the nature of the penalty as a scale-free quadratic form in the
Fourier space. The third requirement is to make the penalty term
consistent with the local kinetic energy contribution coming from the
Thomas-Fermi term.

It is clear that the von Weizs\"acker term in the usual TFDW model is
the unique term consistent with all the relations above. Similarly, it
is then easy to see that in the case of massless relativistic fermions
the unique choice of the von Weizs\"acker-type term for graphene is
given by the first term in \eqref{EE}. Indeed, the first two
requirements above are obviously satisfied, and to check the third
one, we see that
\begin{align}
  & \int_{\RR^2} \left| \nabla^{\frac12} \bigl(\sqrt{\abs{\kappa
        \rho(\lambda x)}} \sgn(\kappa \rho(\lambda x))\bigr) \right|^2
  \ud^2 x \notag \\
  & \qquad \qquad = \kappa \lambda^{-1} \int_{\RR^2} \left|
    \nabla^{\frac12} \bigl(\sqrt{\abs{\rho(x)}} \sgn(\rho(x))\bigr)
  \right|^2 \ud^2  x, \label{eq:TrescW}
\end{align}
and
\begin{align}
  & \int_{\RR^2} |\kappa \rho(\lambda x)|^{3/2}\ud^2 x = \kappa^{3/2}
  \lambda^{-2} \int_{\RR^2} | \rho(x)|^{3/2} \ud^2
  x, \label{eq:TrescTF}
\end{align}
for any $\kappa > 0$ and $\lambda > 0$. Choosing $\kappa \lambda^{-2}
= 1$ to ensure that $\int_{\RR^2} |\kappa \rho(\lambda x)| \ud^2 x =
\int_{\RR^2} |\rho(x)| \ud^2 x$, we have that the right-hand sides of
both \eqref{eq:TrescW} and \eqref{eq:TrescTF} are rescaled by the same
factor. From the dimensional considerations we expect to have
$C_\text{W} \sim \hbar v_F$.

Let us also discuss the presence of $\sgn(\rho)$ in the definition of
the von Weiz\-s\"acker-type term in \eqref{EE}. As will be seen below,
it imparts the energy with some extra degree of symmetry and makes the
energy functional in \eqref{EE} better behaved mathematically, thus
making it a natural modeling choice. Note that this issue is absent in
the conventional TFDW model, since in the case of massive
non-relativistic fermions $\rho$ corresponds to the density of a
single type of charge carriers and is, therefore, non-negative. In any
case, when $\rho \geq 0$, i.e., when the holes are absent from the
consideration, our von Weizs\"acker-type term coincides with one that
has appeared in many studies of relativistic matter and can be further
used to bound at least part of the kinetic energy of electrons from
below \cite{lieb96}.

Another way to understand the origin of the von Weizs\"acker-type term
in the energy is to consider the leading order ``gradient'' correction
to the energy of a uniform system of non-interacting particles. If
\begin{equation}
  \label{eq:T}
  T(\rho) = C_{\text{W}}\int_{\RR^2} \left| \nabla^{\frac12}
    \bigl(\sqrt{\abs{\rho(x)}} \sgn(\rho(x))\bigr) \right|^2 \ud^2 x + \frac23
  C_{\text{TFD}} \int_{\RR^2} \abs{\rho(x)}^{3/2}
  \ud^2 x,
\end{equation}
is the ``kinetic'' part of the energy (recall, however, our discussion
of the exchange and correlation effects above), then the excess
contribution of the kinetic energy to the leading order in $\delta
\rho(x) = \rho(x) - \rho_0$ (i.e., the second variation $\delta^2 T$
of $T$ around $\rho_0$), where $\rho_0 \not= 0$ is the uniform
background density, is
\begin{align}
  \label{eq:T2}
  \delta^2 T = \frac14 C_\text{W} |\rho_0|^{-1} \int_{\RR^2}
  |\nabla^{\frac12} \delta \rho(x)|^2 \ud^2 x + \frac14 C_\text{TFD}
  |\rho_0|^{-1/2} \int_{\RR^2} |\delta \rho(x)|^2 \ud^2 x,
\end{align}
or, in terms of the Fourier transform $\delta \widehat
\rho_\mathbf{k}$ of $\delta \rho(x)$ is given by
\begin{align}
  \label{eq:T2k}
  \delta^2 T = \frac12 \int_{\RR^2} {\ud^2 k \over (2 \pi)^2}
  \Pi_\mathbf{k}^{-1} |\delta \widehat \rho_\mathbf{k} |^2, \qquad
  \Pi_\mathbf{k} = {2 \over C_\text{TFD} |\rho_0|^{-1/2} + C_\text{W}
    |\rho_0|^{-1} |\mathbf k|}.
\end{align}
Here $\Pi_\mathbf{k} = 2 |\rho_0|^{1/2} C_\text{TFD}^{-1} (1 -
C_\text{W} C_\text{TFD}^{-1} |\rho_0|^{-1/2} |\mathbf k|) + O(|\mathbf
k|^2) $ is the polarizability operator for our model. In the absence
of interactions this operator should coincide to the leading order for
$|\mathbf k| \to 0$ with the zero frequency limit of the Lindhard
function of an ideal gas of massless relativistic fermions, and a
comparison is, therefore, in order. The Lindhard function for
non-interacting electrons in graphene was first analyzed by Shung
\cite{shung86} and was later computed in closed form by many authors
\cite{gonzalez94,Ando:06,HwangDasSarma:07} (for a review, see
\cite{KotovUchoa:12}). Restricting the contributions to the
polarizability to only the intraband excitations, one indeed recovers
an expression consistent with the expansion of $\Pi_\mathbf{k}$ in
\eqref{eq:T2k}. However, a peculiar feature of graphene is that when
both the intraband (perturbations of the Fermi surface) and the
interband (formation of virtual electron-hole pairs) excitations are
considered, the intraband and the interband contributions cancel each
other out, making the total polarizability $\Pi_\mathbf{k}^0$ of the
noninteracting massless relativistic fermions independent of $\mathbf
k$ for an interval of $|\mathbf k|$ around zero \cite{KotovUchoa:12}:
\begin{align}
  \label{eq:PiRPA}
  \Pi_\mathbf{k}^0 = {2 |\rho_0|^{1/2} \over \sqrt{\pi} \hbar
    v_F}, \qquad |\mathbf k| \leq 2 \sqrt{\pi |\rho_0|}.
\end{align}
This behavior is due to the cancellation of the contribution from the
two bands of the Dirac cone because of symmetry, as discussed in
\cite{HwangDasSarma:07}.  It is, however, argued (for example in
\cite{Ando:06, WangFertig:11}) that the electron-electron interaction
might lead to breaking this symmetry and changing the asymptotic
behavior so that $\Pi_\mathbf{k}$ decreases linearly near
$\abs{\mathbf k} = 0$. Clearly, correlation effects associated with
Coulombic attraction between electrons and holes should result in a
decreased contribution to the polarizability from the interband
excitations.  This would be consistent with the TFDW model we are
proposing here. Thus we are thinking of the first term in \eqref{EE}
as a non-local contribution of exchange and correlations to an
orbital-free density functional theory beyond the usual local density
approximation. In any case, the model considered here might be viewed
as a natural generic density functional theory model for graphene or
two-dimensional massless relativistic fermions in general.

We finally discuss the rescaling of \eqref{EE} leading to
\eqref{E0}. Introduce
\begin{equation}
  \wt{x} = \lambda x,
  \quad \wt{\rho}(\wt{x}) = \kappa \rho(x),
  \quad \wt{E}(\wt{\rho}) = \gamma E(\rho).
\end{equation}
Then the energy functional in \eqref{EE} becomes
\begin{equation}
  \begin{aligned}
    \frac{1}{\gamma} \wt{E}_0(\wt{\rho}) & =
    \frac{C_{\text{W}}}{\kappa \lambda} \int_{\RR^2} \left|
      \nabla^{\frac12} \bigl(\sqrt{\abs{\wt{\rho}(\tilde x)}}
      \sgn(\wt\rho(\tilde
      x))\bigr) \right|^2 \ud^2 \wt{x} \\
    & \quad + \frac{2}{3} \frac{C_{\text{TFD}}}{\kappa^{3/2}
      \lambda^2} \int_{\RR^2} \abs{\wt{\rho}(\wt x)}^{3/2} \ud^2
    \wt{x} - \frac{Z e^2}{\eps_d \kappa \lambda} \int_{\RR^2}
    \frac{\wt{\rho}(\wt{x})}{(\lambda^2 d^2 + |\wt{x}|^2)^{1/2}}
    \ud^2 \wt{x} \\
    & \quad + \frac{e^2}{2\eps_d \kappa^2 \lambda^3} \iint_{\RR^2
      \times \RR^2} \frac{\wt{\rho}(\wt{x})
      \wt{\rho}(\wt{y})}{\abs{\wt{x} - \wt{y}}} \ud^2 \wt{x} \ud^2
    \wt{y}.
  \end{aligned}
\end{equation}
Taking $\lambda = 1/d$, $\kappa = (\eps_d C_{\text{TFD}} d / e^2 Z)^2$
and $\gamma = \eps_d^3 C_{\text{TFD}}^2 d / (e^ 2 Z)^3$, we arrive at
\eqref{E0} (after dropping tildes) with
\begin{align}
  & a = \frac{\gamma C_{\text{W}} }{\kappa \lambda} =
  \frac{\eps_d C_{\text{W}}}{Z e^2}, \label{aZe} \\
  & b = \frac{\gamma e^2}{\eps_d \kappa^2 \lambda^3} = \frac{Z
    e^4}{\eps_d^2 C_{\text{TFD}}^2}.
\end{align}
Our choice of the rescaling is dictated by the fact that $d$ is the
only length scale for the considered problem, which can be seen from
the fact that the parameters $a$ and $b$ of the rescaled energy are
completely independent of $d$. Also, the units of $\rho$ and $E$ are
now $\kappa^{-1}$ and $\gamma^{-1}$, respectively.

\section{Statement of results}
\label{sec:statement-results}

We start with the energy functional \eqref{eq:E} for a general
background potential $V(x)$, with parameters $a > 0$ and $b > 0$, and
background charge $\bar\rho \in \RR$. Note that since the energy is
invariant with respect to the transformation
  \begin{align}
    \label{eq:rhominusrho}
    \rho \to -\rho, \qquad \bar \rho \to -\bar \rho, \qquad V \to -V,
  \end{align}
  it is sufficient to consider only the case $\bar \rho \geq 0$.

  We point out from the outset that existence of minimizers for the
  energy in \eqref{eq:E} with a general (smooth, decaying) potential
  $V(x)$ is not a priori clear, since the term involving $V(x)$ in
  \eqref{eq:E} may not be bounded from below in the natural function
  classes in which the other terms in the energy are
  well-defined. Nevertheless, if $V(x)$ is of the form of
  \eqref{eq:V}, then it is easy to see that $V \in
  \mathring{H}^{1/2}(\RR^2)$ and, hence, the term involving $V$ in the
  energy can be controlled by the Coulomb repulsion term. Indeed, by
  an explicit computation we have
\begin{align}
  \label{V12}
  (-\Delta)^{1/2} V(x) = -\int_{\RR^3} { |1 + z|  \ud \mu(y, z)
    \over ( |1 + z|^2 + |x - y|^2 )^{3/2} },
\end{align}
implying that $(-\Delta)^{1/2} V(x)$ is smooth and decays no slower
than $|x|^{-3}$ for the considered class of measures $\mu$. Therefore,
in view of the fact that $V(x)$ is smooth and decays no slower than
$|x|^{-1}$, we obtain that
\begin{align}
  \label{V1212}
  \| V \|^2_{\mathring{H}^{1/2}(\RR^2)} = \int_{\RR^2} V
  (-\Delta)^{1/2} V \ud^2 x < \infty.
\end{align}
In fact, our existence results below only rely on the fact that the
estimate in \eqref{V1212} holds. Therefore, throughout the rest of the
paper we generalize the energy in \eqref{eq:E} to potentials $V \in
\mathring{H}^{1/2}(\mathbb R^2)$. We note that by fractional Sobolev
embedding \cite[Theorem 8.4]{Lieb-Loss}, \cite[Theorem
6.5]{palatucci12}, these are functions in $L^4(\mathbb R^2)$, so
the energy $E(\rho)$ in \eqref{eq:E} is well-defined at least for
$\rho - \bar \rho \in C^\infty_c(\mathbb R^2)$.

Caution, however, is necessary in order to assign the meaning to the
energy in \eqref{eq:E} for sufficiently large admissible classes when
searching for minimizers, since the problem is formulated on an
unbounded domain and $\rho - \bar\rho$ does not have a sign a
priori. Indeed, even if the natural classes of functions to consider
would consist of $\rho \in L^1_{loc}(\RR^2)$, it is not a priori clear
if $\rho - \bar\rho$ can be interpreted as a charge density in the
sense of potential theory (i.e., whether
$d \mu = (\rho - \bar \rho) \ud^2 x$ can be associated to a signed
measure $\mu$ on $\RR^2$, making the last term in \eqref{eq:E}
meaningful, see Example \ref{example-nonint}).  The latter depends on
the delicate decay properties of the minimizers and will be the
subject of a separate work \cite{lmm2}. Here we avoid these
difficulties by introducing the induced electrostatic potential $U$
which solves distributionally
\begin{align}
  \label{U}
  (-\Delta)^{1/2} U = \rho - \bar \rho.
\end{align}
We then introduce
\begin{equation}
\begin{aligned}
  \label{E}
  E(\rho) & := a \left\| \sgn(\rho) \sqrt{|\rho|} - \sgn(\bar \rho)
    \sqrt{|\bar \rho|} \right\|_{\mathring{H}^{1/2}(\RR^2)}^2 \\
  & + \int_{\RR^2} \left( \frac{2}{3} \abs{\rho(x)}^{3/2} -
    \frac{2}{3} |\bar\rho|^{3/2} - |\bar\rho|^{1/2} \sgn(\bar\rho)
    ( \rho(x) - \bar \rho) \right) \ud^2 x \\
  & + \langle V, U \rangle _{\mathring{H}^{1/2}(\RR^2)} + \frac{b}{2}
  \left\| U \right\|_{\mathring{H}^{1/2}(\RR^2)}^2 .
\end{aligned}
\end{equation}
Here $\langle \cdot, \cdot \rangle_{\mathring{H}^{1/2}(\RR^2)}$ and
$\| \cdot \|_{\mathring{H}^{1/2}(\RR^2)}$ are the inner product and
the norm associated with the Hilbert space
$\mathring{H}^{1/2}(\RR^2)$, respectively (for details about the
function spaces see Sec. \ref{sec:func}). It is then easy to see that
the definition of $E(\rho)$ in \eqref{E} agrees with that in
\eqref{eq:E} when $\rho - \bar \rho \in C^\infty_c(\RR^2)$. Note that
the second line in \eqref{E} is always non-negative and becomes zero
only for $\rho = \bar \rho$.

We now define the following class of functions for which the energy
$E$ defined in \eqref{E} is meaningful:
\begin{align}
  \label{eq:Am}
  \mathcal A_{\bar\rho} := \Bigg\{ \rho - \bar \rho \in
  \mathring{H}^{-1/2}(\RR^2) : \sgn(\rho) \sqrt{|\rho|} - \sgn(\bar
  \rho) \sqrt{|\bar \rho|} \in \mathring{H}^{1/2}(\RR^2) \Bigg\},
\end{align}
in the sense that $E: \mathcal A_{\bar\rho} \to \mathbb R \cup
\{+\infty\}$.  To see that this class consists of functions and not
merely of distributions, define $u \in \mathring{H}^{1/2}(\RR^2)$ for
a given $\rho \in \mathcal A_{\bar\rho}$ as
\begin{align}
  \label{u}
  u := \sgn(\rho) \sqrt{|\rho|} - \sgn(\bar \rho) \sqrt{|\bar \rho|}.
\end{align}
Then by fractional Sobolev embedding \cite[Theorem 8.4]{Lieb-Loss},
\cite[Theorem 6.5]{palatucci12}, we have $u \in L^4(\RR^2)$ and,
hence, $\rho \in L^2_{loc}(\RR^2)$. In particular, the integral
in the second line in \eqref{E} is locally well-defined.

We begin with a general result on existence of minimizers for $E$ in
\eqref{E} over $\mathcal A_{\bar\rho}$.
\begin{theorem}\label{thm:Egeneral}
  Let $\bar \rho \in \RR$, let $E$ be defined by \eqref{E} with $V \in
  \mathring{H}^{1/2}(\mathbb R^2)$, and let $\inf_{\rho \in \mathcal
    A_{\bar\rho}} E(\rho) < 0$. Then there exists $\rho_0 \in \mathcal
  A_{\bar\rho}$ such that $E(\rho_0) = \inf_{\rho \in \mathcal
    A_{\bar\rho}} E(\rho)$. Furthermore, $\rho_0 \not\equiv \bar
  \rho$, $\rho_0 \in C^{1/2}(\RR^2) \cap L^{\infty}(\RR^2)$ and
  $\rho_0(x) \to \bar \rho$ as $|x| \to \infty$.
\end{theorem}
\noindent
We note that the assumption
$\inf_{\rho \in \mathcal A_{\bar\rho}} E(\rho) < 0$ in Theorem
\ref{thm:Egeneral} is only needed to produce a {\em non-trivial}
minimizer. Otherwise by inspection $\rho = \bar \rho$ is automatically
a minimizer. Thus, existence of minimizers for $E$ over
$\mathcal A_{\bar \rho}$ is guaranteed for every
$V \in \mathring{H}^{1/2}(\RR^2)$.  Also, as a consequence of its
minimizing property, the function $\rho_0(x)$ in Theorem
\ref{thm:Egeneral} solves distributionally the Euler-Lagrange equation
associated with $E$ in \eqref{E}:
\begin{align}
  \label{ELEgen}
  0 = a (-\Delta)^{1/2} \left( \sgn \rho \sqrt{|\rho|} \right) +
  \sqrt{|\rho|} \left( \sgn \rho \sqrt{|\rho|} - \sgn \bar\rho
      \sqrt{|\bar\rho|} + V + b U \right).
\end{align}
In fact, it is more natural to write \eqref{ELEgen} in terms of the
variable $u$ defined in \eqref{u} (see
Sec. \ref{sec:EulerLagrange}). Let us also mention that while H\"older
regularity holds for general potentials $V$ from
$\mathring{H}^{1/2}(\RR^2)$, if $\rho$ changes sign one may not be
able to obtain arbitrarily high regularity of $\rho$ for smooth
potentials $V$ like in \eqref{E0}, see Remark \ref{remark-regularity}.

While the result in Theorem \ref{thm:Egeneral} gives a very general
existence result, it provides only a few basic properties of the
minimizers. In particular, it is not a priori clear whether $\rho_0$
has a sign, even for the potential due to a single charged impurity
appearing in the definition of $E_0$ in \eqref{E0}. This is not merely
a technical issue, since in graphene one generally needs to account
for the presence of both electrons and holes, especially at the
neutrality point, i.e., when $\bar\rho = 0$. It would seem plausible,
however, that in certain situations the minimizers consist only of the
charge carriers of one type. We speculate that this may indeed be the
case for the minimizers of $E_0$ in \eqref{E0} for all values of the
parameters. At least in the asymptotic limits $a \to 0$ or
$b \to \infty$ the minimizers of $E_0$ are expected to be positive.
We caution the reader, however, that in general the situation is
rather delicate, since, even for a negative $V$ with nice decay
properties at infinity, the minimizer might still change sign
\cite{lmm2}.

Motivated by the above observations, for $\bar \rho \geq 0$ we
introduce an admissible class consisting of densities $\rho \geq 0$,
which implies that there are only electrons in the graphene layer:
\begin{align}
  \label{eq:Amp}
  \mathcal A_{\bar\rho}^+ := \Bigl\{ \rho - \bar \rho \in
  \mathring{H}^{-1/2}(\RR^2) : \sqrt{\rho} - \sqrt{\bar\rho} \in
  \mathring{H}^{1/2}(\RR^2), \ \rho \geq 0 \Bigr\}.
\end{align}
Within this admissible class, we have the following counterpart of
Theorem \ref{thm:Egeneral} in the case of strictly positive background
charge.

\begin{theorem}\label{thm:Egenplus}
  Let $\bar \rho >0$, let $E$ be defined by \eqref{E} with $V \in
  \mathring{H}^{1/2}(\mathbb R^2)$
  and let
  $V\not\equiv 0$.  Then there exists a unique $\rho_0 \in \mathcal
  A_{\bar\rho}^+$ satisfying $E(\rho_0) = \inf_{\rho \in \mathcal
    A_{\bar\rho}^+} E(\rho)$. Furthermore, $\rho_0 \not\equiv
  \bar\rho$, $\rho_0 > 0$, $\rho_0 \in C^{1/2}(\RR^2)\cap
  L^\infty(\RR^2)$ and $\rho_0(x) \to \bar\rho$ as $|x| \to \infty$.
\end{theorem}
\noindent
One would naturally expect minimizers in Theorem \ref{thm:Egenplus} to
coincide with the one in Theorem \ref{thm:Egeneral} in many
situations, yet it seems difficult to prove this at the moment.  It is
clear, however, that $\rho_0$ from Theorem \ref{thm:Egenplus} is a
local minimizer of $E$ with respect to smooth perturbations with
compact support. As a consequence, these minimizers solve pointwise
the Euler-Lagrange equation associated with the energy, which for
$\rho > 0$ simplifies to
\begin{align}
  \label{ELEpos}
  0 = a (-\Delta)^{1/2} \left( \sqrt{\rho} \right) + \sqrt{\rho}
  \left( \sqrt{\rho} - \sqrt{\bar\rho} + V + b U \right).
\end{align}
We also note that the assumption of boundedness of $V$ in Theorem
\ref{thm:Egenplus} is needed to ensure strict positivity of the
minimizer, which is required to obtain \eqref{ELEpos}. In addition,
positivity of $\rho_0$ implies further regularity under additional
smoothness assumptions on $V$. In particular, $\rho_0 \in
C^\infty(\RR^2)$ if $V \in C^\infty(\RR^2)$, see Remark
\ref{remark-regularity-plus}.

We note that one of the main differences with the result of Theorem
\ref{thm:Egeneral} in the case of Theorem \ref{thm:Egenplus} is that
there is uniqueness of minimizers, which is due to a kind of strict
convexity of the functional $E$ over $\mathcal A_{\bar\rho}^+$. In
fact, due to this strict convexity one should further expect
uniqueness of solutions of \eqref{ELEpos} and, in particular, that the
minimizer $\rho_0$ in Theorem \ref{thm:Egenplus} is
radially-symmetric, if so is the potential $V$ \cite{lmm2}.

\begin{remark}
  It is easy to see from \eqref{E} that if $V \equiv 0$, the unique
  minimizer of $E$ over $\mathcal A_{\bar\rho}$ is $\rho = \bar
  \rho$. At the same time, if $\bar \rho > 0$ and $\rho = \bar \rho$
  is a minimizer of $E$ over $\mathcal A_{\bar\rho}$, by
  \eqref{ELEpos} we have $V \equiv 0$ and, hence, there are no other
  minimizers. This and the fact that $\mathcal A^+_{\bar \rho} \subset
  \mathcal A_{\bar \rho}$ also implies that if $\bar \rho > 0$, $V \in
  \oH^{1/2}(\RR^2)$ and $V \not\equiv 0$, then by Theorem
  \ref{thm:Egenplus} and the above discussion we also have $\inf_{\rho
    \in \mathcal A_{\bar \rho}} E(\rho) < 0$, i.e., the assumptions of
  Theorem \ref{thm:Egeneral} are satisfied.
\end{remark}

Even though we do not know whether in general the minimizers of $E$
over $\mathcal A_{\bar \rho}$ are positive, in the case of $\bar \rho
> 0$ we are able to prove that this is indeed the case for potentials
$V$ which are, in some sense, ``small''. The smallness of the
potential is expressed in terms of the magnitude of its
$\mathring{H}^{1/2}(\mathbb R^2)$ norm. Our result is given by the
following theorem.

\begin{theorem} \label{thm:positive} Let $\bar \rho >0$, let $E$ be
  defined by \eqref{E} with $V \in \mathring{H}^{1/2}(\mathbb R^2)$
  and let $V\not\equiv 0$.  Then there exists a constant $C > 0$
  depending only on $a$, $b$ and $\bar\rho$ such that if $\| V
  \|_{\mathring{H}^{1/2}(\mathbb R^2)} \leq C$, then the unique
  minimizer $\rho_0 > 0$ of $E$ over $\mathcal A_{\bar\rho}^+$ in
  Theorem \ref{thm:Egenplus} coincides with the minimizer of $E$ over
  $\mathcal A_{\bar \rho}$ in Theorem \ref{thm:Egeneral}.
\end{theorem}

\noindent We note that in the parameter regime of Theorem
\ref{thm:positive} the minimizer $\rho_0 > 0$ does not deviate much
from $\bar \rho > 0$. In particular, if
$\| V \|_{\mathring{H}^{1/2}(\mathbb R^2)} \to 0$, one expects to
recover, to the leading order, the solution of \eqref{ELEpos}
linearized around $\rho = \bar \rho$, which expresses the linear
response of the system to the perturbation by the potential $V$ and
describes {\em screening} of the external charge by free electrons in
the graphene layer. A more detailed analysis of this phenomenon will
be carried out in the forthcoming paper \cite{lmm2}. Note that within
the Thomas-Fermi type models of the usual electron systems screening
was studied mathematically in \cite{Lieb-Simon-1977} for the
Thomas-Fermi model and in \cite{Cances-2011} for the
Thomas-Fermi-von~Weizs\"acker model.

We now focus on the main situation of physical interest, in which the
layer is at the neutrality point. In particular, we wish to
investigate how a graphene layer reacts to external charges in the
presence of a supply of electrons from a lead at infinity. Fixing
$\bar \rho = 0$, we know that under the assumptions of Theorem
\ref{thm:Egeneral} there is a non-trivial minimizer in the class
$\mathcal A_0$. As we already mentioned, we do not know whether this
minimizer also belongs to $\mathcal A_0^+$, even for a potential
defined in \eqref{eq:V} with a positive measure $\mu$. Nevertheless,
if we restrict the admissible class to $\mathcal A_0^+$, we have the
following analog of Theorem \ref{thm:Egenplus}.

\begin{theorem}\label{thm:EA0p}
  Let $\bar \rho = 0$, let $E$ be defined by \eqref{E} with $V \in
  \mathring{H}^{1/2}(\mathbb R^2)$,
  and let
  $\inf_{\rho \in \mathcal A_{\bar\rho}^+} E(\rho) < 0$. Then there
  exists a unique $\rho_0 \in \mathcal A_{\bar\rho}^+$ satisfying
  $E(\rho_0) = \inf_{\rho \in \mathcal A_{\bar\rho}^+}
  E(\rho)$. Furthermore, $\rho_0 > 0$, $\rho_0 \in C^{1/2}(\RR^2) \cap
  L^{\infty}(\RR^2)$ and $\rho_0(x) \to 0$ as $\abs{x} \to \infty$.
\end{theorem}

Let us point out that, in contrast to Theorem \ref{thm:Egenplus}, the
condition that $V \not\equiv 0$ is not sufficient for existence of
non-trivial minimizers in Theorem \ref{thm:EA0p}. In fact, it can be
shown, following the arguments in the proof of Theorem
\ref{thm:positive} that for sufficiently small values of
$\| V \|_{\mathring{H}^{1/2}(\mathbb R^2)}$ the energy $E$ in
\eqref{eq:E} cannot have non-trivial minimizers. We illustrate this
point by considering the case of the energy $E_0$ in \eqref{E0}, which
is also of particular interest because of its physical significance.
Defining
\begin{align}
  \label{ac}
  a_c := {\Gamma^2(\frac14) \over 2 \Gamma^2(\frac34)},
\end{align}
where $\Gamma(x)$ is the Gamma function and $a_c \approx 4.3769$ is
the inverse of the Hardy constant for the operator square root of the
negative Laplacian \cite[Remark 4.2]{frank08}, we have the following
result for the generalization of the energy $E_0$ in \eqref{E0}.

\begin{theorem}\label{thm:E0A0}
  Let $\bar \rho = 0$ and let $E$ be defined by \eqref{E} with
  \begin{equation}
    \label{V0}
    V(x) = -\frac{1}{(1 + |x|^2)^{1/2}}.
  \end{equation}
  Then:
  \begin{enumerate}[(i)]
  \item If $a \geq a_c$, then $\rho_0 = 0$ is the unique minimizer of
    $E$ over $\mathcal A_0$.
  \item If $a < a_c$, then there exists a minimizer $\rho_0 \not\equiv
    0$ of $E$ over $\mathcal A_0^+$.
  \end{enumerate}
\end{theorem}
\noindent Thus, for $a$ sufficiently large (or, equivalently, for the
impurity valence $Z$ sufficiently small or the effective dielectric
constant $\epsilon_d$ sufficiently large, see \eqref{aZe}) there can
be no bound states between the charge carriers in graphene and a
single charged impurity. In other words, this implies a surprising
result that for $a \geq a_c$ the charged impurity elicits no response
from the electrons in the graphene layer (within the considered
density functional theory). The bifurcation at $a = a_c$ is determined
by a fine balance between the first term in the energy and the
potential term, which has the same asymptotics when $|x| \to \infty$
as the Hardy potential for $(-\Delta)^{1/2}$.

Note that the statement of Theorem \ref{thm:E0A0} obviously remains
true if $\mathcal A_0^+$ is replaced with $\mathcal A_0$. Also note
that the magnitude of $b$ does not play any role for existence
vs. non-existence of non-trivial minimizers in this case.  At the same
time, as we will show in the forthcoming paper \cite{lmm2}, both the
values of $a$ and $b$, together with the (finite) $L^1$ norm of the
minimizer $\rho_0 \in \mathcal A_0^+$ determine the algebraic rate of
decay of $\rho_0(x)$ as $|x| \to \infty$. Specifically, we expect
\begin{align}
  \label{rho0dec}
  \rho_0(x) \sim {1 \over |x|^{2s}}, \qquad |x| \to \infty,
\end{align}
where $s \in (1,2)$ is the unique solution of the algebraic equation
\begin{align}
  \label{sdec}
  {2 a \Gamma ({s + 1 \over 2}) \Gamma( {2 - s \over 2} ) \over
    \Gamma( {1 - s \over 2} ) \Gamma( {s \over 2} ) } = 1 - {b \over 2
    \pi} \| \rho_0(x) \|_{L^1(\RR^2)} ,
\end{align}
which is formally obtained by linearizing \eqref{ELEpos} with respect
to $\sqrt{\rho}$, using the leading order asymptotics of $V$ and $U$
in the far field and looking for distributional solutions in the form
appearing in \eqref{rho0dec}. This prediction is confirmed by the
results of the numerical solution of \eqref{ELEpos}. Figure \ref{fig0}
shows the solution of \eqref{ELEpos} for $a = 1$ and $b = 1$ (we refer
to \cite{lmm2} for further details), for which we found
$\|\rho_0\|_{L^1(\RR^2)} \simeq 6.95$ and
$\rho_0(x) \simeq 0.28 |x|^{-2.2}$ for $|x| \gg 1$. This agrees well
with \eqref{sdec}.  Thus, in contrast to previous studies, our model
predicts a non-trivial dependence of the algebraic decay rate of the
positive mininimizers on the parameters. Note that since for
$s \in (1,2)$ the term multiplying $a$ in \eqref{sdec} is negative, we
have $\| \rho_0(x) \|_{L^1(\RR^2)} > 2 \pi b^{-1}$. In the original
physical variables it means that the total charge induced in the
graphene layer {\em exceeds} in absolute value the external
out-of-plane charge. Note that this is similar to what is observed in
the Thomas-Fermi-von Weizs\"acker model of a single atom
\cite{benguria81}.

\begin{figure}[t]
  \centering
  \includegraphics[width=2.3in]{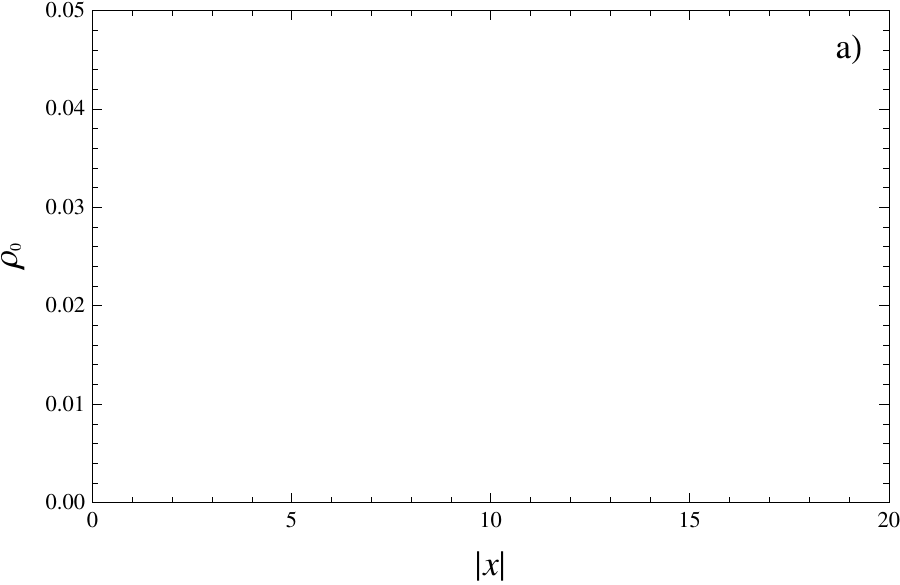}
  \hspace{3mm} \includegraphics[width=2.26in]{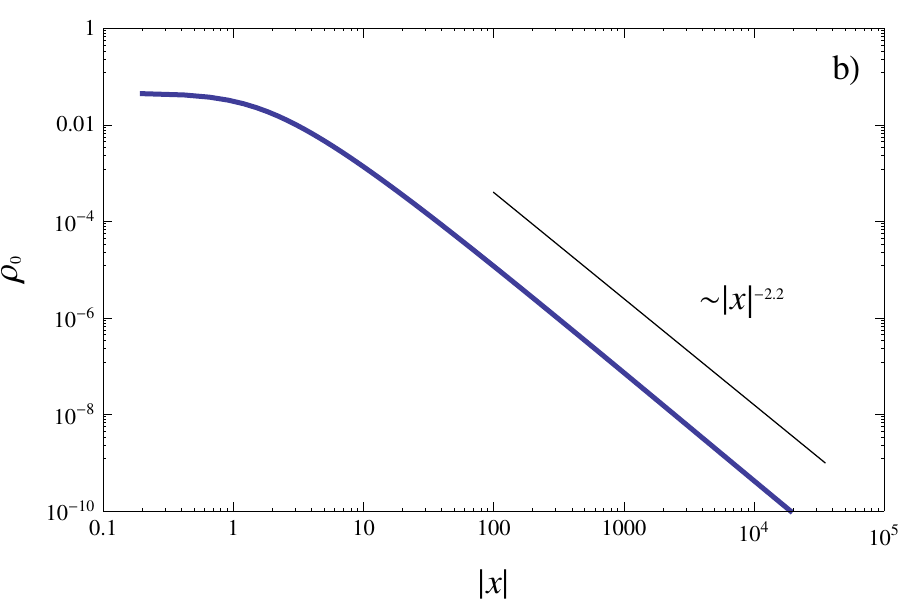}
  \caption{The minimizer $\rho_0$ in Theorem \ref{thm:E0A0} for $a =
    1$ and $b = 1$, plotted on a linear (a) and logarithmic (b)
    scale. }
  \label{fig0}
\end{figure}

\section{Preliminaries}
\label{sec:prelim}

\subsection{Functional setting}
\label{sec:func}

Recall that the homogeneous Sobolev space $\mathring{H}^{1/2}(\RR^2)$
can be defined as the completion of $C^\infty_c(\RR^2)$ with respect
to the Gagliardo's norm
\begin{align}\label{Gagliardo}
  \|u\|_{\oH^{1/2}(\RR^2)}^2 := \frac{1}{4\pi} \iint_{\RR^2 \times
    \RR^2}\frac{|u(x)-u(y)|^2}{|x-y|^{3}}\ud^2 x \ud^2 y.
\end{align}
By Plancherel's identity (cf. \cite[Lemma 3.1]{frank08}), on
  $C^\infty_c(\mathbb R^2)$ the
$\|\cdot\|_{\mathring{H}^{1/2}(\RR^2)}$--norm admits an equivalent
Fourier representation
\begin{align}
  \|u\|_{\mathring{H}^{1/2}(\RR^2)}^2 = \int_{\RR^2} | \, | \mathbf
  k|^{1/2} \hat u_\mathbf{k}|^2 {\ud^2 k \over (2 \pi)^2}, \qquad \hat
  u_\mathbf{k} = \int_{\RR^2} e^{i \mathbf k \cdot x} u(x) \ud^2 x,
\end{align}
which suggests the notation
\begin{align}\label{half-grad}
  \|u\|_{\mathring{H}^{1/2}(\RR^2)}^2 =:
  \int_{\RR^2}|\nabla^{\frac12}u(x)|^2\ud^2 x,
\end{align}
which we often use in this paper.  By the fractional Sobolev
inequality \cite[Theorem 8.4]{Lieb-Loss}, \cite[Theorem
6.5]{palatucci12},
\begin{align}\label{sobolev-inequality}
  \|u\|_{\mathring{H}^{1/2}(\RR^2)}^2\ge \sqrt{\pi} \,
  \|u\|_{L^4(\RR^2)}^2,\qquad\forall u\in C^\infty_c(\RR^2).
\end{align}
In particular, the space $\mathring{H}^{1/2}(\RR^2)$ is a well-defined
space of functions and
\begin{align}\label{sobolev-embedd}
\mathring{H}^{1/2}(\RR^2)\subset L^4(\RR^2).
\end{align}
The space $\mathring{H}^{1/2}(\RR^2)$ is also a Hilbert space, with
the scalar product associated to \eqref{Gagliardo} given by
\begin{align}\label{bi-Gagliardo}
  \langle u,v\rangle_{\oH^{1/2}(\RR^2)} := \frac{1}{4\pi}\iint_{\RR^2
    \times \RR^2} \frac{(u(x)-u(y))(v(x)-v(y))}{|x-y|^{3}}\ud^2 x
  \ud^2 y.
\end{align}
The dual space to $\mathring{H}^{1/2}(\RR^2)$ is denoted
$\mathring{H}^{-1/2}(\RR^2)$.  According to the Riesz representation
theorem, for every $F\in\mathring{H}^{-1/2}(\RR^2)$ there exists a
uniquely defined {\em potential} $v\in\mathring{H}^{1/2}(\RR^2)$ such
that
\begin{align}\label{weak-potential}
  \langle v,\varphi\rangle_{\oH^{1/2}(\RR^2)}=\langle
  F,\varphi\rangle\qquad\forall\varphi\in \mathring{H}^{1/2}(\RR^2),
\end{align}
where $\langle F,\cdot\rangle:\mathring{H}^{1/2}(\RR^2)\to\RR$ denotes
the bounded linear functional generated by $F$.  Moreover,
\begin{align}\label{isometry}
  \|v\|_{\oH^{1/2}(\RR^2)}=\|F\|_{\mathring{H}^{-1/2}(\RR^2)},
\end{align}
so the duality \eqref{weak-potential} is an isometry.
The potential $v\in \mathring{H}^{1/2}(\RR^2)$ satisfying
\eqref{weak-potential} is interpreted as the {\em weak solution} of
the linear equation
\begin{align}\label{weak-Laplace}
(-\Delta)^{1/2}v= F\qquad\text{in $\RR^2$}.
\end{align}

Recall that for functions $u\in C^\infty(\RR^2) \cap L^1(\mathbb
  R^2, (1 + |x|)^{-3} \ud^2 x)$, the fractional Laplacian
$(-\Delta)^{1/2}$ can be defined as
\begin{align}\label{half-Laplace}
  (-\Delta)^{1/2}u(x)= \frac{1}{4\pi}\int_{\RR^2}\frac{2 u(x) - u(x+y)
    - u(x-y)}{|y|^3}\ud^2 y\qquad(x\in\RR^2).
\end{align}
Note that the second order Taylor expansion of function $u$
yields that the strong singularity of the integrand at the origin is
removed, and \eqref{half-Laplace} can be understood as a converging
Lebesgue integral, see \cite[Lemma 3.2]{palatucci12}.  Of course, the
weighted second order differential quotient in \eqref{half-Laplace}
coincides with a more standard definition of $(-\Delta)^{1/2}$ as a
pseudodifferential operator, in the sense that for all $u\in
C^\infty_c(\RR^2)$,
\begin{align}
  \widehat{\left( (-\Delta)^{1/2}u \right)_\mathbf{k}} = |\mathbf
    k| \hat u_\mathbf{k},
\end{align}
cf. \cite[Proposition 3.3]{palatucci12}. In particular, this makes
  the definition of $(-\Delta)^{1/2}$ in \eqref{half-Laplace}
  consistent with the notation used in \eqref{weak-Laplace}.

Note that if $u\in C^\infty_c(\RR^2)$ then $(-\Delta)^{1/2}u\in
C^\infty(\RR^2)$, but is not compactly supported and in fact,
\begin{equation}\label{asymptotic-cube}
(-\Delta)^{1/2}u=O(|x|^{-3})\qquad\text{as $|x|\to\infty$},
\end{equation}
see \cite[Lemma 1.2]{mazja}.  In particular, this shows that the
operator $(-\Delta)^{1/2}$ could be extended by duality to the
weighted space $L^1(\RR^2,(1+|x|)^{-3}\ud^2 x)$, that is for $u\in
L^1(\RR^2,(1+|x|)^{-3}\ud^2 x)$,
\begin{align}
  \langle(-\Delta)^{1/2}u,\varphi\rangle=\int_{\RR^2}u(x)(-\Delta)^{1/2}\varphi(x)
  \ud^2 x\qquad\forall\varphi\in C^\infty_c(\RR^2)
\end{align}
and this definition agrees with \eqref{half-Laplace} in the case
$u\in C^\infty_c(\RR^2)$, see \cite[p. 73]{Silvestre:07}.  Clearly,
$\mathring{H}^{1/2}(\RR^2)\subset L^1(\RR^2,(1+|x|)^{-3}\ud^2 x)$.  In
particular, this implies that for $v\in\mathring{H}^{1/2}(\RR^2)$,
\begin{align}\label{distributional-half-Laplace}
  \langle
  v,\varphi\rangle_{\oH^{1/2}(\RR^2)}=\int_{\RR^2}v(x)(-\Delta)^{1/2}\varphi(x)\ud^2
  x\qquad\forall\varphi\in C^\infty_c(\RR^2).
\end{align}

When $f\in C^\infty_c(\RR^2)$, the left inverse to $(-\Delta)^{1/2}$
is represented by the {\em Riesz potential}, i.e., if $u$ is the
  weak solution of $(-\Delta)^{1/2}u=f$ then $u$ admits the integral
representation
\begin{align}\label{inverse-half-Laplace}
u(x)=(-\Delta)^{-1/2}f(x)=\frac{1}{2\pi}\int_{\RR^2}\frac{f(y)}{|x-y|}\ud^2 y,
\end{align}
see \cite[Lemma 1.3]{mazja}. Such integral representation could be
extended to a wider class of functions and (signed) measures,
cf. \cite[Lemma 1.8, 1.11]{mazja}.  In particular, taking $f =
\delta(x)$, we obtain that $1/(2\pi|x|)$ is the fundamental solution
of $(-\Delta)^{1/2}$.  We emphasize, however, that not every potential
of a linear functional $f\in \mathring{H}^{-1/2}(\RR^2)$ admits an
integral representation \eqref{inverse-half-Laplace}.  Similarly, not
every linear functional $f\in \mathring{H}^{-1/2}(\RR^2)$ admits an
integral representation of the norm in terms of the Coulomb energy.
If $f\in L^1_{loc}(\RR^2)$ satisfies
\begin{align}
  \iint_{\RR^2 \times \RR^2} \frac{|f(x)||f(y)|}{|x-y|}\ud^2 x \ud^2
  y<+\infty.
\end{align}
then $f\in \mathring{H}^{-1/2}(\RR^2)$ in the sense that
\begin{align}
  \langle f,\varphi\rangle:=\int_{\RR^2} f(x) \varphi(x) \ud^2 x
\end{align}
is a bounded linear functional on $\mathring{H}^{1/2}(\RR^2)$ and the
norm of $\langle f,\cdot\rangle$ is expressed in terms of the Coulomb
energy
\begin{align}\label{Coulomb-plus}
  \|f\|_{\mathring{H}^{-1/2}(\RR^2)}^2 = \frac{1}{2\pi}\iint_{\RR^2
    \times \RR^2} \frac{f(x)f(y)}{|x-y|}\ud^2 x \ud^2 y,
\end{align}
see e.g. \cite[pp. 96-97]{mazja}. In particular, from Sobolev
inequality \eqref{sobolev-inequality} we conclude by duality that
\begin{align}\label{sobolev-embedd-dual}
L^{4/3}(\RR^2)\subset\mathring{H}^{-1/2}(\RR^2)
\end{align}
and \eqref{Coulomb-plus} is valid for every $f\in L^{4/3}(\RR^2)$.
But at the same time, one could construct a sequence of sign--changing
functions $\{f_n\}\subset C^\infty_c(\RR^2)$ such that $\{f_n\}$ is a
Cauchy sequence in $\mathring{H}^{-1/2}(\RR^2)$, but $\{f_n\}$ does
not converge a.e. to a measurable function or more generally, to a
(signed) measure on $\RR^2$.  See \cite{armitage75,rempel76} or
\cite[Theorem 1.19]{landkof}, \cite[p. 97]{duPlessis} for other
relevant examples which go back to H. Cartan \cite[Remark 13 on
p. 87]{Cartan}. Below we present a different example which involves
smooth functions, rather than measures like in Cartan's type examples.

\begin{example}\label{example-nonint}
Define
\begin{align}
u_a(x_1,x_2) = a^{1/2}\exp(-|x|^2)\cos(a x_1).
\end{align}
Then, using Fourier transform, we can calculate that
\begin{equation}
  \norm{u_a}_{\oH^{-1/2}(\RR^2)}^2 = \frac{\sqrt{2}}{8}\pi^{3/2} a
  e^{-\frac{a^2}{2}} \left(e^{\frac{a^2}{4}} I_0\bigl(
    \tfrac{1}{4}a^2\bigr)+1\right),
\end{equation}
where $I_0(z)$ is the modified Bessel function of the first
kind. Taking the limit $a \to \infty$, one gets
\begin{align}
  \lim_{a \to \infty} ||u_a||_{\oH^{-1/2}(\RR^2)}^2 = {\pi \over 4}.
\end{align}
A Cauchy sequence in $\oH^{-1/2}(\RR^2)$ that fails to converge to a
signed measure can then be constructed as
\begin{align}
u_n(x_1,x_2) = \sum_{k=1}^n e^{k/4}\exp(-|x|^2)\cos(e^k x_1).
\end{align}
Since this series is dominated in $\oH^{-1/2}(\RR^2)$ by a geometric
series, it converges in $\oH^{-1/2}(\RR^2)$.  But clearly it does not
converge to a signed measure.
\end{example}

\subsection{Hardy--Littlewood--Sobolev and H\"older estimates}

We recall the well-known Hardy--Littlewood--Sobolev \cite[Theorem 1 in
Section V.1.2]{stein70} and H\"older estimates on the Riesz potentials
of functions in $L^p(\RR^2)$.  Surprisingly, we were not able to find
a concise reference to H\"older estimate, although the result is
standard.  Instead, we refer to \cite[Theorem 5.2]{gatto}, where the
estimate is obtained in an abstract framework of fractional integral
operators.

\begin{lemma}\label{lemma-HLS}
Let $f\in L^s(\RR^2)$ for some $s\in(1,2)$ and
\begin{align}\label{Riesz}
v(x)= \frac{1}{2\pi}\int_{\RR^2}\frac{f(y)}{|x-y|}\ud^2 y\qquad(x\in\RR^2).
\end{align}
Then $v\in L^t(\RR^2)$ with $\frac{1}{t}=\frac{1}{s}-\frac{1}{2}$
and
\begin{align}
  \label{e-Lt}
  \| v \|_{L^t(\RR^2)} \leq C \| f \|_{L^s(\RR^2)},
\end{align}
for some $C > 0$ depending only on $s$.  Furthermore, if $f\in
L^s(\RR^2)\cap L^1(\RR^2,(1+|x|)^{-1}\ud^2 x)$ for some $s>2$, then
$v\in L^\infty(\RR^2)\cap C^{1-{2 \over s}}(\RR^2)$ and
\begin{equation}\label{e-Holder}
  |v(x)-v(y)|\le C\|f\|_{L^s(\RR^2)}|x-y|^{1-{2 \over s}}\qquad\forall x,y\in\RR^2,
\end{equation}
for some $C > 0$ depending only on $s$.
\end{lemma}

\begin{remark}
  The assumption $f\in L^1(\RR^2,(1+|x|)^{-1}\ud^2 x)$ in the second
  part of the lemma is a necessary and sufficient condition which
  ensures that $|v(x)|<+\infty$ a.e. in $\RR^2$, assuming that the
  operator in \eqref{Riesz} is understood in the (Lebesgue) integral
  sense, c.f. \cite[(1.3.10) on p. 61]{landkof}.  Observe that by
  H\"older inequality all the assumptions of the second part of
  Lemma \label{r-HLS} \ref{lemma-HLS} are satisfied, if $f \in
  L^s(\RR^2)$ for all $s \in [s_1, s_2]$ for some $1 < s_1 < 2 < s_2 <
  \infty$.
\end{remark}

\subsection{Interior regularity}

We are going to show that although $(-\Delta)^{1/2}$ is a nonlocal
operator, the interior regularity of solutions of \eqref{weak-Laplace}
does not depend on the behavior of the right-hand side at infinity.
The proof of this basic fact can be found in \cite[Proposition
2.22]{Silvestre:07}. Here, however, we give a quantitative version of
the above statement.

\begin{lemma}\label{lemma-harmonic}
  Let $f\in L^1_{loc}(\RR^2)$, let $p \geq 1$ and let $u\in
  L^p(\RR^2)$ be such that
  \begin{align}\label{Euler-0}
    \langle u,\varphi\rangle_{\oH^{1/2}(\RR^2)}=\int_{\RR^2}
    f(x)\varphi(x)\ud^2 x\qquad\forall \varphi\in C^\infty_c(\RR^2).
  \end{align}
  Assume that $f=0$ on $B_{2R}(0)$ for some $R > 0$. Then $u\in
  C^\infty(\bar B_R(0))$ and for every $n \geq 0$
  \begin{equation}
  \| \nabla^n u \|_{L^{\infty}(B_R(0))}  \leq C R^{-n-\frac{2}{p}} \|u \|_{L^p(\RR^2)}
  \end{equation}
  for some $C > 0$ depending only on $n$ and $p$.
\end{lemma}

\proof
Let $\eta_R(x) = \eta(|x| / R)$, where $\eta \in C^\infty(\RR)$ is a
smooth cut-off function such that $\eta(x)=1$ for all $|x| > 2$,
$\eta(x)=0$ for all $|x| < \frac32$, and $0\le\eta\le 1$.  Given
$\varphi\in C^\infty_c(\RR^2)$ supported on $B_R(0)$, let $\psi\in
  \oH^{1/2}(\RR^2)$ be a weak solution of $(-\Delta)^{1/2} \psi =
  \varphi$. By \eqref{inverse-half-Laplace} we have
\begin{align}
  \psi(x) =
  \frac{1}{2\pi}\int_{\RR^2}\frac{\varphi(y)}{|x-y|}\ud^2 y,
\end{align}
and, in particular, $\psi \in C^\infty(\mathbb R^2)$. Then
$(1-\eta_R)\psi\in C^\infty_c(\RR^2)$ and is supported on
$B_{2R}(0)$. Testing \eqref{Euler-0} with $(1-\eta_R)\psi$ and taking
into account \eqref{distributional-half-Laplace}, we obtain
\begin{align}\label{EL-eta}
  0=\langle u,(1-\eta_R)\psi\rangle_{\oH^{1/2}(\RR^2)}  =
    \int_{\mathbb R^2}
    u(x)  (-\Delta)^{1/2} \Big( (1 - \eta_R)\psi \Big) (x) \ud^2 x \notag \\
  =\int_{B_R(0)}u(x)\varphi(x)\ud^2 x - \int_{\mathbb R^2} u(x)
    (-\Delta)^{1/2} (\eta_R\psi)(x) \ud^2 x.
\end{align}

Inserting the definition of $(-\Delta)^{1/2}$ from
\eqref{half-Laplace} and changing the order of integration in the
  last integral in \eqref{EL-eta} yields
\begin{align}
    & 
    \int_{\RR^2}
    u(x)(-\Delta)^{1/2}(\eta_R \psi)(x)\ud^2 x \notag \\
    & =
    \frac{1}{8\pi^2}\int_{\RR^2}u(x)\int_{\RR^2}|y|^{-3}\int_{B_R(0)}
    \left(\frac{2\eta_R(x)}{|x-z|}
      - \frac{\eta_R(x+y)}{|x+y-z|}\right.  \notag \\
    & \hspace{17em} -
    \left.\frac{\eta_R(x-y)}{|x-y-z|}\right)\varphi(z)\ud^2 z \ud^2 y
    \ud^2 x  \notag \\
    & = \frac{1}{8\pi^2}\int_{B_R(0)} \int_{\RR^2}
    \left(\int_{\RR^2}|y|^{-3}\left(
        \frac{2\eta_R(x)}{|x-z|} - \frac{\eta_R(x+y)}{|x+y-z|}
      \right.\right.  \notag \\
    & \hspace{17em} -
    \left.\left.\frac{\eta_R(x-y)}{|x-y-z|}\right)\ud^2
      y\right)u(x) \varphi(z)  \ud^2 x \ud^2 z  \notag \\
    & = \int_{B_R(0)} \int_{\RR^2} J_R(x,z)u(x)\varphi(z)\ud^2 x \ud^2
    z,
\end{align}
where for $x\in\RR^2$ and $z \in B_R(0)$ we introduced
\begin{align}
  J_R(x,z):=\frac{1}{8\pi^2}\int_{\RR^2} |y|^{-3}
  \left(\frac{2\eta_R(x)}{|x-z|} - \frac{\eta_R(x+y)}{|x+y-z|}-
    \frac{\eta_R(x-y)}{|x-y-z|} \right)\ud^2 y.
\end{align}
Observe that
\begin{equation}
  J_R(x, z) =(-\Delta)^{1/2}_x j_R(x, z), \qquad j_R(x,
  z):= \frac{1}{2\pi} \frac{\eta_R(x)}{|x-z|}.
\end{equation}
Clearly, $j_R(x, z)=0$ for $x\in B_{3R/2}(0)$ and $j_R \in
C^\infty(\RR^2 \times \bar B_R(0))$, with
\begin{align}
  & |\nabla^n_z j_R(x, z)| \leq c_n (R + \abs{x-z})^{-(n+1)}, \\
  & |\nabla^2_x \nabla^n_z j_R(x, z)| \leq c_n R^{-2} (R + \abs{x-z})^{-(n+1)}
\end{align}
for all $n\ge 0$ and some $c_n >0$ (unless stated otherwise, all
  constants in this proof depend only on $n$ and the choice of
  $\eta$). Then
\begin{align}\label{eq:estgradx}
  & \abs{\nabla^n_z J_R(x,z)} \leq \frac{1}{8\pi^2}\int_{\RR^2}
  |y|^{-3} \Big\lvert 2 \nabla^n_z j_R(x, z) - \nabla^n_z j_R(x+y,
  z) \notag \\
  & \hspace{18em} -
  \nabla^n_z j_R(x-y, z)  \Big\rvert \ud^2 y \notag\\
  & = \frac{1}{8\pi^2} \int_{B_R(0)} |y|^{-3} \Big| 2 \nabla^n_z
  j_R(x, z) - \nabla^n_z j_R(x+y, z)
  - \nabla^n_z j_R(x-y, z) \Big| \ud^2 y \notag\\
  & \quad + \frac{1}{8\pi^2} \int_{\RR^2 \backslash B_R(0)} |y|^{-3} \Big| 2
  \nabla^n_z j_R(x, z) - \nabla^n_z j_R(x+y, z) -
  \nabla^n_z j_R(x-y, z) \Big| \ud^2 y \notag\\
  & \leq \frac{1}{8\pi^2} \norm{\nabla^2_x \nabla^{n}_z j_R(\cdot,
    z)}_{L^{\infty}(\RR^2)} \int_{B_R(0)} \abs{y}^{-1} \ud^2 y \notag\\
  & \quad + \frac{1}{2\pi^2} \norm{\nabla^n_z j_R(\cdot,
    z)}_{L^{\infty}(\RR^2)}
  \int_{\RR^2 \backslash B_R(0)} |y|^{-3} \ud^2 y \notag\\
  & \leq C_n R^{-n-2},
\end{align}
for some $C_n > 0$.  In particular, for any $x \in \RR^2$, $J_R(x,
\cdot) \in C^\infty(\bar B_R(0))$.

We next prove that for some $c_n>0$ we have
\begin{equation}
  \label{J1x3}
  |\nabla_z^n J_R(x,z)|\le \frac{c_n}{R^{n-1} (R^3 + |x|^{3})}\qquad\forall
  x\in\RR^2,\quad
  \forall z\in\bar B_R(0) .
\end{equation}
For $\abs{x} \leq 4R$, the estimate follows from \eqref{eq:estgradx}.
Now assume $|x| \geq 4R$. Then $\eta_R(x)=1$, and since
$1/(2\pi\abs{x})$ is the fundamental solution for $(-\Delta)^{1/2}$,
we have for $\abs{z} \leq R$
\begin{equation}
  \frac{1}{8\pi^2} \int_{\RR^2} |y|^{-3} \left(\frac{2}{|x-z|} -
    \frac{1}{|x+y-z|} - \frac{1}{|x-y-z|}\right) \ud^2 y=0.
\end{equation}
Using this fact we can rewrite
\begin{align}
  J_R(x,z) & =
  \frac{1}{8\pi^2}\int_{\RR^2}\frac{1-\eta_R(x-y)}{|x-y-z|}|y|^{-3}\ud^2
  y +
  \frac{1}{8\pi^2}\int_{\RR^2}\frac{1-\eta_R(x+y)}{|x+y-z|}|y|^{-3}\ud^2
  y \notag \\
  & = \frac{1}{8\pi^2}\int_{\RR^2}\frac{1}{|y-z|} \left(
    \frac{1-\eta_R(y)}{\abs{x-y}^3} + \frac{1-\eta_R(y)}{\abs{x +
        y}^3} \right) \ud^2 y \notag \\
  & =: \frac{1}{8\pi^2}\int_{\RR^2}\frac{1}{|y-z|} h_R(x, y) \ud^2 y.
\end{align}
Notice that for fixed $x$ with $\abs{x} \geq 4R$, $h_R(x, \cdot) \in
C^{\infty}_c(\RR^2)$ and its support is contained in
$B_{2R}(0)$. Therefore,
\begin{equation}\label{eq:rieszderiv}
  \begin{aligned}
    \abs{\nabla^n_z J_R(x, z)}
    & \leq \frac{1}{8 \pi^2} \int_{B_{2R}(0)}
    \frac{1}{\abs{y - z}} \abs{\nabla^n_y h_R(x, y)} \ud^2 y.
  \end{aligned}
\end{equation}
For $y \in B_{2R(0)}$ and $\abs{x} \geq 4R$, we have the estimate
$\abs{\nabla^n_y h_R(x, y)}
\leq C_n R^{-n} \abs{x}^{-3}$ for some $C_n > 0$. Therefore,
\begin{equation}
  |\nabla_z^n J_R(x,z)|\le \frac{C_n}{R^n |x|^3}
  \int_{B_{2R}(0)}\frac{1}{|y-z|}\ud^2 y \le \frac{C_n'}{R^{n-1}|x|^3},
\end{equation}
for some $C_n' > 0$.

Finally, taking into account \eqref{EL-eta}, we conclude that for
almost every $z\in \bar B_R(0)$ we have
\begin{align}
  u(z)=\int_{\RR^2} J_R(x,z)u(x)\ud^2 x,
\end{align}
and, since \eqref{J1x3} leads to $\|\nabla_z^n J_R(\cdot, z) \|_{L^{p
    \over p -1}(\RR^2)} \leq C R^{-n-2/p}$ for some $C > 0$ depending
only on $n$, $p$ and the choice of $\eta$, the statement of the lemma
follows by H\"older inequality. \qed

\section{Variational setting}
\label{sec:var}

\subsection{A representation of the energy
  functional} \label{sec:renormenergy}

Recall that for a given $\rho \in \mathcal A_{\bar\rho}$ we define $u$
by
\begin{align}\label{u-rho}
  u := \sgn(\rho) \sqrt{|\rho|} - \sgn(\bar \rho) \sqrt{|\bar \rho|}
\end{align}
and set $\bar u := \sqrt{\abs{\bar\rho}}\sgn{\bar\rho}$.  Then $u \in
\mathring{H}^{1/2}(\RR^2)$ in view of the definition of $\mathcal
A_{\bar\rho}$.  Since $\sgn(\rho) \sqrt{|\rho|}=u+\bar u$, we can
define
\begin{multline} \label{psi-abs}
  \int_{\RR^2}\big|\nabla^{\frac12}\bigl(\sqrt{\abs{\rho(x)} }
  \sgn(\rho(x))\bigr)\big|^2\ud^2 x :=\\
  \frac{1}{4\pi}\iint_{\RR^2 \times \RR^2} \frac{|(u(x) + \bar
    u)-(u(y)+\bar u)|^2}{|x-y|^{3}}\ud^2 x \ud^2 y=
  \|u\|_{\oH^{1/2}(\RR^2)}^2,
\end{multline}
which justifies and clarifies the notation used in Sections 1--3 of
the paper.

\begin{figure}[t]
  \centering
  \includegraphics[width=2.3in]{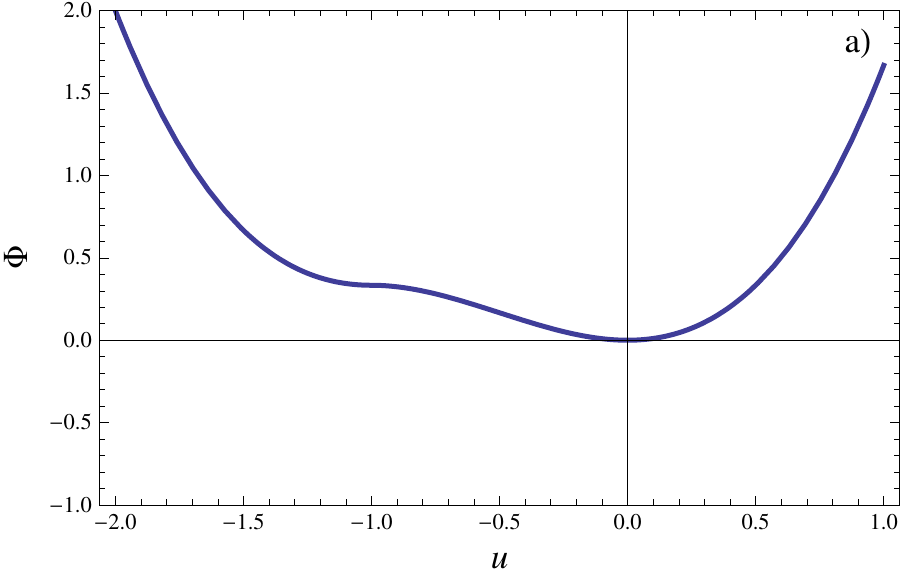}
  \hspace{3mm} \includegraphics[width=2.26in]{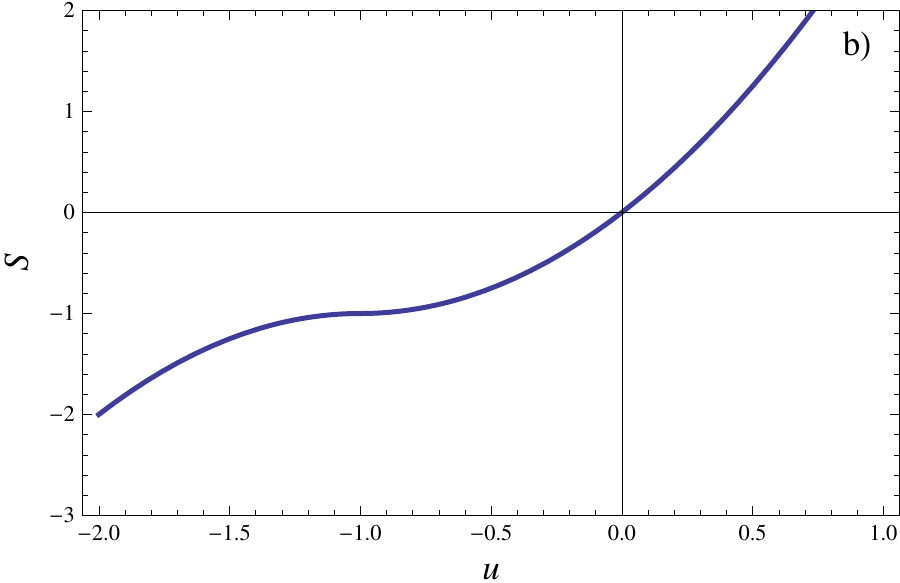}
  \caption{(a) Plot of $\Phi(u)$ and (b) plot of $S(u)$ for $\bar
      u = 1$.}
  \label{fig1}
\end{figure}

Throughout the rest of the paper we assume, without loss of
generality, that $\bar \rho \geq 0$, and, hence, $\bar u \geq 0$ (see
\eqref{eq:rhominusrho}). Denote
\begin{align}\label{S-def}
  S(u):=|u+\bar u|(u+\bar u)-|\bar u|\bar u =
  \begin{cases}
    2 \bar u u + u^2, & u \geq -\bar u, \\
    -u^2 - 2 \bar u u - 2 \bar u^2, & u < -\bar u,
  \end{cases}
\end{align}
and
\begin{align}\label{Phi}
  \Phi(u):=\frac{2}{3}(|u+\bar u|^3 - |\bar u|^3 ) - \bar u S(u)
  =\left\{
\begin{array}{ll}
  \hspace{0.75em}\frac{2}{3}u^3+\bar u u^2,& u\ge-\bar u,\smallskip\\
  -\frac{2}{3}u^3-\bar u u^2+\frac{2}{3}\bar u^3,& u<-\bar u.
\end{array}
\right.
\end{align}
The graphs of $\Phi(u)$ and $S(u)$ for $\bar u = 1$ are presented in
Fig. \ref{fig1}. Clearly $S,\Phi\in C^1(\RR)$ and both functions are
smooth functions of $u \in \RR$ except at $u = -\bar u$.  Moreover
\begin{align}
  \label{Phi-abs}
  c(\bar u |u|^2+|u|^3)\le\Phi(u)\le C(\bar u
  |u|^2+|u|^3)\qquad(u\in\RR),
\end{align}
\begin{align}
  \label{S-abs}
  c(\bar u |u|+|u|^2)\le S(u)\sgn(u)\le C(\bar u
  |u|+|u|^2)\qquad(u\in\RR),
\end{align}
for some universal $C > c > 0$. Therefore, for
$u \in C^{\infty}_c(\RR^2)$, the energy $E(u)$ can be written as (with
a slight abuse of notation, we use the same letter to denote both the
energy as a function of $\rho$ and that as a function of $u$ in the
rest of the paper)
\begin{multline}
  \label{eq:Eu}
  E(u) = a \|u\|_{\oH^{1/2}(\RR^2)}^2
  + \int_{\RR^2} \Phi(u(x)) \ud^2 x \\
  + \int_{\RR^2} V(x) S(u(x))\ud^2 x + \frac{b}{2}
  \iint_{\RR^2 \times \RR^2} \frac{S(u(x))S(u(y)) }{\abs{x - y}}  \ud^2 x \ud^2 y.
\end{multline}
Given $u\in\mathring{H}^{1/2}(\RR^2)$, \eqref{sobolev-embedd} and
\eqref{S-abs} imply that $S(u)\in L^2_{loc}(\RR^2)$. Then for all
$\varphi\in C^\infty_c(\RR^2)$ we can define
\begin{align}\label{Su-linear}
\langle S(u),\varphi\rangle: = \int_{\RR^2}S(u(x))\varphi(x) \ud^2 x.
\end{align}
We say $S(u)\in \mathring{H}^{-1/2}(\RR^2)$, if the linear functional
$\langle S(u),\cdot\rangle$ defined in \eqref{Su-linear} is bounded
by a multiple of $\|\varphi\|_{\mathring{H}^{1/2}(\RR^2)}$.  In
that case $\langle S(u),\cdot\rangle$ is understood as the unique
continuous extension of \eqref{Su-linear} to
$\mathring{H}^{1/2}(\RR^2)$.  Note that $S(u)\in
\mathring{H}^{-1/2}(\RR^2)$ does not necessarily imply that $S(u)w\in
L^1(\RR^2)$ for every $w\in \mathring{H}^{1/2}(\RR^2)$.  In other
words, $\langle S(u),\cdot\rangle$ does not always admit an integral
representation on $\mathring{H}^{1/2}(\RR^2)$, as observed by Brezis
and Browder in \cite{brezis79} in the context of $H^1(\RR^N)$.

\subsection{Class $\mathcal H$}
Introduce the class
\begin{align}
  \label{eq:H}
  \mathcal H:= \Bigl\{ u \in \mathring{H}^{1/2}(\RR^2): S(u) \in
  \mathring{H}^{-1/2}(\RR^2) \Bigr\}.
\end{align}
As discussed in Section~\ref{sec:renormenergy}, this is an equivalent
way of writing the class $\mc{A}_{\bar{\rho}}$.  Given $u\in \mathcal
H$, Riesz's representation theorem uniquely defines a potential
$U_{S(u)}\in\mathring{H}^{1/2}(\RR^2)$ such that
\begin{align}\label{Riesz-repr}
  \langle {U}_{S(u)},\varphi\rangle_{\oH^{1/2}(\RR^2)}=\langle
  S(u),\varphi\rangle\qquad\forall\varphi\in
  \mathring{H}^{1/2}(\RR^2).
\end{align}

\noindent In particular, from the Sobolev embedding
\eqref{sobolev-embedd} combined with \eqref{Phi-abs} we obtain the
following inclusions:
\begin{align}\label{L2}
  & \{ u \in \mathcal{H} \ : \ E(u) < +\infty \} \subset
  L^4(\RR^2)\cap L^2(\RR^2)\qquad\text{if $\bar u\neq
    0$,}\\
  \label{L3} & \{ u \in \mathcal{H} \ : \ E(u) < +\infty \} \subset
  L^4(\RR^2)\cap L^3(\RR^2)\qquad\text{if $\bar u= 0$}.
\end{align}

\begin{remark}
  In fact, using a fractional extension of the Brezis-Browder argument
  in \cite{brezis79}, one can establish stronger inclusions:
\begin{align}
  \label{L2+}
  & \mathcal{H} \subset L^4(\RR^2)\cap L^2(\RR^2)\qquad\text{if $\bar
    u\neq 0$,}\\
  \label{L3+}
  & \mathcal{H} \subset L^4(\RR^2)\cap L^3(\RR^2)\qquad\text{if $\bar
    u= 0$}.
\end{align}
We refer to the forthcoming work \cite{lmm2} for the details. Moreover,
these inclusions are, in some sense optimal. To see
the optimality of \eqref{L2+}, choose $u\in
C^\infty_c(B_1(0))$, a vector $e\in\RR^2$ with $|e|=1$ and for
$N\in\mathbb N$ let
  \begin{align}\label{uN}
    u_N(x):=\frac{1}{\sqrt{N}}\sum_{k=1}^N u\big(x+k\exp(N)e\big).
  \end{align}
  It is standard to check (cf. \eqref{asymptotic-cube} for the
  $\mathring{H}^{1/2}$--term and \cite[p. 363]{ruiz10} for the Coulomb
  term) that
  \begin{align}
    \|u_N\|_{\mathring{H}^{1/2}(\RR^2)}\simeq\|S(u_N)\|_{\mathring{H}^{-1/2}(\RR^2)}
    \simeq C,
  \end{align}
  while
  \begin{align}
    \|u_N\|_{L^p(\RR^2)}=O(N^{\frac{1}{p}-\frac{1}{2}}).
  \end{align}
  We conclude that the sequence $\{u_N\}$ is not bounded in
  $L^p(\RR^2)$ for any $p<2$.  To check the optimality of \eqref{L3+},
  instead of \eqref{uN} one can use an appropriately rescaled family
  of functions $u_N$, similar to those in \cite[Proof of Theorem
  1.5]{ruiz10}.
\end{remark}

\section{Proof of Theorems \ref{thm:Egeneral} and \ref{thm:positive}}
\label{sec:th31}

\subsection{Existence of a minimizer}
If $V\in \mathring{H}^{1/2}(\RR^2)$ then we can rewrite
$E$ in terms of $u$ and the associated potential ${U}_{S(u)}$ as
\begin{equation}\label{E-v}
E(u) = a \|u\|_{\oH^{1/2}(\RR^2)}^2
  + \int_{\RR^2} \Phi(u(x)) \ud^2 x   + \langle V,{U}_{S(u)}\rangle_{\oH^{1/2}(\RR^2)} +
  \frac{b}{2} \|{U}_{S(u)}\|_{\oH^{1/2}(\RR^2)}^2.
\end{equation}
In particular, it is easy to see that
\begin{align}
-\frac{1}{2b} \|V\|_{\oH^{1/2}(\RR^2)}^2\le
  \inf_{u \in \mathcal H} E(u) \le 0.
\end{align}
We are going to prove that $E$ attains a minimizer on $\mathcal H$.

\begin{proposition}\label{Frechet}
  If $V\in \mathring{H}^{1/2}(\RR^2)$ then there exists $u_0 \in
  \mathcal H$ such that $E(u_0) = \inf_{u \in \mathcal H} E(u)$.
\end{proposition}

\proof Consider a minimizing sequence $\{u_n\}\subset\mathcal H$ and
the corresponding sequence of potentials $\{{U}_{S(u_n)}\}\subset
\mathring{H}^{1/2}(\RR^2)$ from \eqref{Riesz-repr}.  Clearly,
\begin{align}
\sup_n \|u_n\|_{\oH^{1/2}(\RR^2)}^2 &\leq C,\\
\sup_n \|{U}_{S(u_n)}\|_{\oH^{1/2}(\RR^2)}^2 &\leq C,
\end{align}
Hence, we may extract subsequences, still denoted by $\{u_n\}$ and
$\{{U}_{S(u_n)}\}$ such that
  \begin{align}
    & u_n \wconv u_0 &\text{in }
    \mathring{H}^{1/2}(\RR^2), \\
    & {U}_{S(u_n)}\wconv v_0 &\text{in }
    \mathring{H}^{1/2}(\RR^2),
  \end{align}
  for some $u_0, v_0\in \mathring{H}^{1/2}(\RR^2)$. Using a fractional
  version of Rellich-Kondrachov theorem \cite[Corollary
  7.2]{palatucci12}, we conclude that
\begin{align}
  & u_n \to u_0 \qquad \text{in } L^{p}_{loc}(\RR^2) \quad \text{for
    all} \quad 1 \leq p<4, \label{Sob}
\end{align}
and, upon extraction of another subsequence, that $u_n(x) \to u_0(x)$ for a.e. $x \in \RR^2$.
Using \eqref{Sob}, \eqref{S-abs} and strong continuity of $S$ as a
Nemytskii operator from $L^p_{loc}(\RR^2)$ into $L^q_{loc}(\RR^2)$
with $q\le p/2$ (cf. \cite[Theorem C.1]{struwe90}), we also
conclude that
\begin{align}\label{S-Sob}
  & S(u_n) \to S(u_0) \qquad \text{in } L^{q}_{loc}(\RR^2)
  \quad\text{for all} \quad 1 \leq q<2.
\end{align}
Using \eqref{Riesz-repr}, \eqref{Su-linear} and \eqref{S-Sob},
similarly to an argument in the proof of \cite[Proposition
2.4]{ruiz10}, for every fixed $\varphi\in C^\infty_c(\RR^2)$ we obtain
\begin{multline}\label{converge}
  \langle v_0,\varphi\rangle_{\oH^{1/2}(\RR^2)} \leftarrow
  \langle {U}_{S(u_n)},\varphi\rangle_{\oH^{1/2}(\RR^2)}=\langle S(u_n),\varphi\rangle\\
  =\int_{\RR^2} S(u_n(x))\varphi(x) \ud^2 x\to\int_{\RR^2}
  S(u_0(x))\varphi(x) \ud^2 x.
\end{multline}
Therefore,
\begin{equation}\label{Riesz-repr-v0}
  \langle v_0,\varphi\rangle_{\oH^{1/2}(\RR^2)}
  =\int_{\RR^2} S(u_0(x))\varphi(x)  \ud^2 x,\qquad
  \forall\varphi\in C^\infty_c(\RR^2).
\end{equation}
Note that $\langle v_0,\cdot\rangle_{\oH^{1/2}(\RR^2)}$ is a bounded
linear functional on $\mathring{H}^{1/2}(\RR^2)$, since
$v_0\in\mathring{H}^{1/2}(\RR^2)$.  Therefore
$S(u_0)\in\mathring{H}^{-1/2}(\RR^2)$. In particular, this means that
$u_0\in\mathcal H$ and
\begin{equation}\label{Riesz-repr-v0-plus}
v_0={U}_{S(u_0)}.
\end{equation}
We conclude that
\begin{multline}\label{E-vv}
  E(u_0) = a \|u_0\|_{\oH^{1/2}(\RR^2)}^2 + \int_{\RR^2} \Phi(u_0(x))
  \ud^2 x \\ + \langle V,{U}_{S(u_0)}\rangle_{\oH^{1/2}(\RR^2)} +
  \frac{b}{2} \|{U}_{S(u_0)}\|_{\oH^{1/2}(\RR^2)}^2
  \le\liminf_{n\to\infty} E(u_n).
\end{multline}
This follows from the weak lower semicontinuity of the norm
$\|\cdot\|_{\oH^{1/2}(\RR^2)}$, continuity of the linear
functional $\langle V,\cdot\rangle_{\oH^{1/2}(\RR^2)}$ on
$\mathring{H}^{1/2}(\RR^2)$, and from the non-negativity of the
function $\Phi$ which allows to apply Fatou lemma in the integral term
which contains $\Phi$.  \qed

\subsection{Euler--Lagrange equation} \label{sec:EulerLagrange}

In order to derive the Euler--Lagrange equation for $E$, we first
establish three auxiliary lemmas.

\begin{lemma}
  \label{uht-admissible}
  Let $u \in \mathcal H$ and $h \in C^{\infty}_c(\RR^2)$. Then $u + t
  h \in \mathcal H$ for every $t \in \mathbb R$.
\end{lemma}

\proof Since obviously $u + t h \in \mathring{H}^{1/2}(\RR^2)$, it
remains to prove that $S(u + t h) \in
\mathring{H}^{-1/2}(\RR^2)$. Consider $F(x) := S(u(x) + t h(x)) -
S(u(x))$. Clearly, $F$ has compact support, and by \eqref{S-abs} we
have $F \in L^2(\RR^2)$. Therefore, we also have $F \in
L^{4/3}(\RR^2)$, and, hence, by \eqref{sobolev-embedd-dual} the
functional
\begin{align}
  \label{Suth}
  \langle F, \varphi \rangle := \int_{\RR^2} (S(u(x) + t h(x)) -
  S(u(x))) \varphi(x) \ud^2 x \qquad (\varphi \in C^\infty_c(\RR^2))
\end{align}
can be continuously extended to the whole of
$\mathring{H}^{1/2}(\RR^2)$. Thus $S(u + t h) - S(u) \in
\mathring{H}^{-1/2}(\RR^2)$, and since $S(u) \in
\mathring{H}^{-1/2}(\RR^2)$ by assumption, this completes the
proof. \qed

\begin{lemma}\label{l-diff-V}
  Let $u\in\mathcal H$ and $h \in C^{\infty}_c(\RR^2)$. Then
  $S'(u)h\in \mathring{H}^{-1/2}(\RR^2) \cap L^4(\RR^2)$, and for
  every $\varphi\in \mathring{H}^{1/2}(\RR^2)$,
\begin{equation}\label{derivative}
  \lim_{t\to 0}\frac{1}{t}\langle S(u+th)-S(u),\varphi
  \rangle=\int_{\RR^2} S'(u(x))h(x)\varphi(x)\ud^2 x.
\end{equation}
\end{lemma}

\proof Note that
\begin{equation}
  \label{S-abs-prime}
  S'(u)=2|u+\bar u|,
\end{equation}
and, hence, $S'(u)\in L^4_{loc}(\RR^2)$ by
\eqref{sobolev-embedd}. Therefore, in view of the fact that $h \in
C^\infty_c(\RR^2)$, we have $S'(u) h \in L^4(\RR^2) \cap
L^{4/3}(\RR^2)$ and, again, by \eqref{sobolev-embedd-dual}, this
implies that $S'(u) h \in \mathring{H}^{-1/2}(\RR^2)$.

At the same time, by the argument in the proof of Lemma
\ref{uht-admissible} we have an integral representation
\begin{equation}\label{int-repr}
  \langle S(u+th)-S(u),\varphi\rangle
  =\int_{\RR^2} \big(S(u(x)+th(x))-S(u(x))\big) \varphi(x)\ud^2 x
\end{equation}
for every $\varphi\in \mathring{H}^{1/2}(\RR^2)$.  Using
\eqref{int-repr} and the mean value theorem, for some
$\theta(t,\cdot)\in L^\infty(\RR^2)$ with
$\|\theta(t,\cdot)\|_{L^\infty}\le 1$, we obtain
\begin{multline}\label{meanvalue}
  \frac{1}{t}\langle S(u+th)-S(u),\varphi\rangle
  =\frac{1}{t}\int_{\RR^2} \big(S(u(x)+th(x))-S(u(x))
  \big)\varphi(x)\ud^2 x\\
  =\int_{\RR^2} S'\big(u(x)+t\theta(t,x)h(x)\big)h(x)\varphi(x)\ud^2
  x,
\end{multline}
where the latter integral converges, since $S'(u+t\theta(t,\cdot)h)\in
L^4_{loc}(\RR^2)$ in view of \eqref{S-abs-prime} and
\eqref{sobolev-embedd}.  Then \eqref{derivative} follows by the Lebesgue
dominated convergence.
\qed

\begin{lemma}\label{l-diff-C}
  Let $u\in\mathcal H$ and $h\in C^\infty_c(\RR^2)$. Then
\begin{equation}
  \lim_{t\to
    0}\frac{1}{t}\left(\|{U}_{S(u+th)}\|_{\oH^{1/2}(\RR^2)}^2-
    \|{U}_{S(u)}\|_{\oH^{1/2}(\RR^2)}^2\right)
  =2\int_{\RR^2} {U}_{S(u)}(x)S'(u(x))h(x)\ud^2 x.
\end{equation}
\end{lemma}

\proof Since $S(u + t h) \in \oH^{-1/2}(\RR^2)$ by Lemma
\ref{uht-admissible}, the potential $U_{S(u + h t)} \in \oH^{1/2}(\RR^2)$ is well-defined.  Then using \eqref{Riesz-repr} we
obtain
\begin{align}\label{diff-quadratic}
  &
  \|{U}_{S(u+th)}\|_{\oH^{1/2}(\RR^2)}^2-\|{U}_{S(u)}\|_{\oH^{1/2}(\RR^2)}^2
  \notag \\
  & \quad =2\langle S(u+th)-S(u),{U}_{S(u)}\rangle+ \langle
  S(u+th)-S(u),{U}_{S(u+th)-S(u)}\rangle.
\end{align}
Similarly to \eqref{meanvalue}, for some $\theta(t,\cdot)\in
L^\infty(\RR^2)$ with $\|\theta(t,\cdot)\|_{L^\infty}\le 1$ we obtain
\begin{multline}\label{meanvalue-plus}
  \frac{1}{t}\left|\langle S(u+th)-S(u),{U}_{S(u+th)-S(u)}\rangle\right|\\
  =\frac{1}{t}\left|\int_{\RR^2}
    \big(S(u(x)+th(x))-S(u(x))\big){U}_{S(u+th)-S(u)}(x)\ud^2
    x\right|\\
  \le
  C\|S'(u+t\theta(t,\cdot)h)h\|_{L^{4/3}(\RR^2)}\|{U}_{S(u+th)-S(u)}\|_{\oH^{1/2}(\RR^2)},
\end{multline}
for some $C > 0$ independent of $t$.  Since $h$ is compactly
supported, by Lebesgue dominated convergence we conclude that
\begin{align}
  &
  \|S'(u+t\theta(t,\cdot)h)h\|_{L^{4/3}(\RR^2)}\to\|S'(u)h\|_{L^{4/3}(\RR^2)}\quad
  \text{as $t\to 0$}, \\
\label{eq:43conv}
& \|S(u+th)-S(u)\|_{L^{4/3}(\RR^2)}\to 0\quad\text{as $t\to 0$}.
\end{align}
{}From \eqref{isometry} we note that
$U:\mathring{H}^{-1/2}(\RR^2)\mapsto\mathring{H}^{1/2}(\RR^2)$ is an
isometry.  Then \eqref{eq:43conv} and \eqref{sobolev-embedd-dual}
imply that
\begin{equation}
\|{U}_{S(u+th)-S(u)}\|_{\oH^{1/2}(\RR^2)}\to 0\quad\text{as $t\to 0$}.
\end{equation}
Using \eqref{derivative} we obtain
\begin{multline}
  \lim_{t\to 0}\frac{1}{t}\left(2\langle
    S(u+th)-S(u),{U}_{S(u)}\rangle+
    \langle S(u+th)-S(u),{U}_{S(u+th)-S(u)}\rangle\right)\\
  =2\int_{\RR^2} S'(u(x))h(x){U}_{S(u)}(x) \ud^2 x.
\end{multline}
Hence, the assertion follows via \eqref{diff-quadratic}.
\qed

\begin{proposition}\label{Euler}
  Let $V\in \mathring{H}^{1/2}(\RR^2)$.  Then $E$ at every
  $u\in\mathcal H$ admits a directional derivative with respect to
  test functions $h\in C^\infty_c(\RR^2)$. Furthermore, the derivative
  is given by
\begin{multline}
  \label{eq:EL-weak}
  \frac{d}{dt}E(u+th)\Big|_{t=0} = 2a \langle u,h
  \rangle_{\oH^{1/2}(\RR^2)} +\int_{\RR^2}\Phi'(u(x))h(x)\ud^2 x\\+
  \int_{\RR^2}V(x)S'(u(x))h(x) \ud^2 x+ b\int_{\RR^2}
  {U}_{S(u)}(x)S'(u(x))h(x) \ud^2 x.
\end{multline}
\end{proposition}

\proof Follows from Lemmas \ref{uht-admissible}--\ref{l-diff-C}
and \eqref{Phi}.  \qed

\begin{remark}\label{Frechet-rem}
  The corresponding Euler--Lagrange equation is then in the
  distributional sense
\begin{align}
  \label{eq:EL+}
  0 = 2a (-\Delta)^{1/2}u +\Phi'(u)+V S'(u)+b{U}_{S(u)}S'(u).
\end{align}
Observing that $\Phi'(u)=uS'(u)$ and $S'(u)=2|u+\bar u|$, we rewrite
\eqref{eq:EL+} in the form
\begin{align}
  \label{eq:EL++}
  0 = a (-\Delta)^{1/2}u +|u+\bar u|\left(u+V+b{U}_{S(u)}\right).
\end{align}
\end{remark}

\subsection{Regularity}\label{sec:regularity}

Using the Euler--Lagrange equation for $E$ we shall establish
additional regularity of the minimizers.

\begin{lemma}\label{lem:regularity}
  Assume that $V\in \mathring{H}^{1/2}(\RR^2)$. Let $u \in \mathcal H$ be
  such that $E(u) = \inf_{\tilde u \in \mathcal H} E(\tilde u)$.  Then
  $u\in C^{1/2}(\RR^2)\cap L^\infty(\RR^2)$ and $u(x) \to 0$ as $|x|
  \to \infty$.
\end{lemma}

\proof Since $u\in \mathcal H$ is a minimizer of $E$, it satisfies the
Euler-Lagrange equation \eqref{eq:EL++} distributionally.  Denote
\begin{equation}
  \label{F}
  F(x):=-a^{-1} |u(x)+\bar
  u|\left(u(x)+V(x)+b{U}_{S(u)}(x) \right), \qquad x \in \RR^2.
\end{equation}
If $F\in L^s(\RR^2)$ for some $1<s<2$ then
\begin{align}\label{Riesz-0}
  u(x)= \frac{1}{2\pi}\int_{\RR^2}\frac{F(y)}{|x-y|}\ud^2 y \in
  L^t(\RR^2),\quad\frac{1}{t}=\frac{1}{s}-\frac{1}{2},
\end{align}
see Lemma \ref{lemma-HLS}.  So we can apply the bootstrap argument in
an attempt to improve the $L^t$--regularity of $u$.

\smallskip

First, we consider the case $\bar u=0$. Then $u \in L^p(\RR^2)$ for
all $p\in[3,4]$, by \eqref{L3}.  Since $V, U_{S(u)}\in L^4(\RR^2)$, we
conclude that
\begin{align}
u^2&\in L^s(\RR^2)\qquad\forall s\in\big[\tfrac{3}{2},2\big],\\
uV,\, u U_{S(u)}&\in L^s(\RR^2)\qquad\forall s\in\big[\tfrac{12}{7},2\big].
\end{align}
Then
\begin{equation}
F\in L^s(\RR^2)\qquad\forall s\in\big[\tfrac{12}{7},2\big],
\end{equation}
and therefore, by \eqref{Riesz-0} and \eqref{L3}
\begin{equation}\label{iter-1}
u\in L^t(\RR^2),\qquad\forall t\ge 3.
\end{equation}
Iterating once more, we deduce that
\begin{equation}\label{s-boot}
F\in L^s(\RR^2)\qquad\forall s\in\big[\tfrac{12}{7},4\big).
\end{equation}
Then by Lemma \ref{lemma-HLS} and Remark \ref{r-HLS} we obtain
\begin{equation}\label{iter-2}
  u\in C^{1-{2 \over s}}(\RR^2)\cap L^t(\RR^2),\qquad\forall t\in[3,\infty], \;
  \forall s \in \bigl(2,4\bigr).
\end{equation}
In particular, this means that in \eqref{s-boot} we can take $s=4$.
Applying Lemma \ref{lemma-HLS} once again with $s=4$, we finally
deduce that
\begin{equation}\label{iter-3}
u\in C^{1/2}(\RR^2)\cap L^t(\RR^2),\qquad\forall t\in[3,\infty].
\end{equation}

\smallskip

Next consider the case $\bar u\neq 0$.  Then $u \in L^p(\RR^2)$ for
all $p\in[2,4]$, by \eqref{L2}.  Since $V, U_{S(u)}\in L^4(\RR^2)$, we
conclude that
\begin{align}
  u^2&\in L^s(\RR^2)\qquad\forall s\in\big[1,2\big],\\
  uV,\, u U_{S(u)}&\in L^s(\RR^2)\qquad\forall
  s\in\big[\tfrac{4}{3},2\big],\\
  \bar uV,\, \bar u U_{S(u)}&\in L^4(\RR^2).
\end{align}
Hence
\begin{equation}
  F = F_1 + F_2, \qquad  F_1 \in L^4(\RR^2), \qquad F_2 \in L^2(\RR^2),
\end{equation}
and we do not gain at this point any additional regularity because of
the lack of decay at infinity coming from $\bar uV$ and $\bar u
U_{S(u)}$.  Since the Riesz potential in \eqref{Riesz-0} could be
applied (as an integral operator) only to functions in $L^s(\RR^2)$
with $s<2$, the previous bootstrap procedure fails on the whole of
$\RR^2$.  Instead, we will use a localized version based on Lemma
\ref{lemma-harmonic}.

Given arbitrary $R>0$, we represent
\begin{equation}
  \label{uRhR}
  u = u_R + h_R, \qquad u_R := u_{R,1} + u_{R,2},
\end{equation}
where
\begin{align}\label{Riesz-0R}
  u_{R,1}(x) :=
  \frac{1}{2\pi}\int_{B_{2R}(0)}\frac{F_1(y)}{|x-y|}\ud^2 y, \qquad
  u_{R,2}(x) :=
  \frac{1}{2\pi}\int_{B_{2R}(0)}\frac{F_2(y)}{|x-y|}\ud^2 y.
\end{align}
Since $\chi_{B_{2R}(0)}F_1 \in L^s(\RR^2)$ for any $s\in[1,4]$, by
H\"older inequality we conclude that
\begin{align}
  \|u_{R,1} \|_{L^\infty(B_R(0))} \leq C_R \| F_1 \|_{L^4(\RR^2)},
\end{align}
for some $C_R > 0$ depending only on $R$ (here and in the rest of the
proof we suppress the dependence of all the constants on $a$, $b$ and
$\bar \rho$).  Similarly, since $\chi_{B_{2R}(0)} F_2 \in L^s(\RR^2)$
for any $s \in [1, 2]$, by Lemma \ref{lemma-HLS} we have $u_{R,2} \in
L^t(\RR^2)$ for all $t > 2$. Furthermore, by H\"older inequality we
obtain
\begin{align}
  \| u_{R,2} \|_{L^t(B_R(0))} \leq C_{R,t} \| F_2 \|_{L^2(\RR^2)},
\end{align}
for some $C_{R,t} > 0$ depending only on $R$ and $t$. At the same
time, the function $h_R := u - u_R$ solves
\begin{align}
  \langle h_R, \varphi \rangle_{\oH^{1/2}(\RR^2)} = \int_{\RR^2
    \backslash B_{2R}(0)} F(x) \varphi(x) \ud^2 x \qquad \forall
  \varphi \in C^\infty_c(\RR^2).
\end{align}
Therefore, by Lemma \ref{lemma-harmonic} we have $h_R\in
W^{1,\infty}(B_R(0))$, with $\| h_R \|_{L^\infty(B_R(0))} \leq C_R \|
u \|_{L^4(\RR^2)}$ for some $C_R > 0$ depending only on $R$. Thus, we
have $u \in L^t(B_R(0))$ for any $t > 2$, with the norm controlled by
constant depending only on $R$, $t$, $\|u\|_{L^4(\RR^2)}$,
$\|V\|_{L^4(\RR^2)}$ and $\|U_{S(u)}
\|_{L^4(\RR^2)}$.
Furthermore, by possibly increasing the value of the constant, we can
make the same conclusion about $\| u \|_{L^t(B_{2R}(0))}$.
Bootstrapping this information, we then obtain that $\chi_{B_{2R}(0)}
F_2 \in L^s(\RR^2)$ with any $s \in [1, 4)$, and, again, by H\"older
inequality this implies that $u_{R,2} \in L^\infty(B_R(0))$, with the
norm controlled by $\|u\|_{L^4(\RR^2)}$, $\|V\|_{L^4(\RR^2)}$ and
$\|U_{S(u)} \|_{L^4(\RR^2)}$, and the constant depending only on
$R$. Combining this with the $L^\infty$-bounds on $u_{R,1}$ and $h_R$,
we then conclude that $\| u \|_{L^\infty(B_R(0))} \leq C_R$ for some
constant $C_R > 0$ depending only on $R$ and $\|u\|_{L^4(\RR^2)}$,
$\|V\|_{L^4(\RR^2)}$ and $\|U_{S(u)} \|_{L^4(\RR^2)}$. Furthermore,
since the obtained estimates for fixed $R > 0$ are translationally
invariant, we arrive at the conclusion that $u \in L^\infty(\RR^2)$.

The fact that $u \in L^\infty(\RR^2)$ implies that $F\in L^4(\RR^2)$.
Noting that $\chi_{B_{2R}(0)} F \in L^s(\RR^2)$ for any $s \in [1,4]$,
by Lemma \ref{lemma-HLS} and Remark \ref{r-HLS} we then have
\begin{equation}\label{e-Holder-plus}
  |u_R(x)-u_R(y)|\le C \| F \|_{L^4(\RR^2)} |x-y|^{1/2}\qquad\forall
  x,y\in\RR^2,
\end{equation}
for some universal $C > 0$. On the other hand, since $\| h_R
\|_{W^{1,\infty}(B_R(0))} \to 0$ as $R \to \infty$, fixing $x$ and $y$
and passing to the limit we conclude that
\begin{equation}\label{e-Holder-plus+}
  |u(x)-u(y)|\le C \| F \|_{L^4(\RR^2)} |x-y|^{1/2}\qquad\forall
  x,y\in\RR^2,
\end{equation}
and, hence,
\begin{equation}\label{iter-3plus}
  u\in C^{1/2}(\RR^2)\cap L^t(\RR^2)\qquad\forall t\in[2,\infty].
\end{equation}
Finally, it is standard to see that $u\in C^{\alpha}(\RR^2)\cap
L^p(\RR^2)$ for some $\alpha\in(0,1]$ and some $p\ge 1$ implies that
$u(x) \to 0$ as $|x| \to \infty$.  \qed

\begin{remark}\label{remark-regularity}
  The regularity of minimizers of $E$ can be improved under additional
  smoothness assumptions on $V$.  For instance, assume that $V\in
  \mathring{H}^{1/2}(\RR^2)\cap C^{1/2}(\RR^2)$.  Taking into account
  that $S(\cdot)$ is a $C^1$--mapping and using Lemma
  \ref{lemma-harmonic}, similarly to the arguments in the proof of
  Lemma~\ref{lem:regularity} one can show that $U_{S(u)} \in
  C^{1/2}(\RR^2) \cap L^\infty(\RR^2)$ as well. Then the
  expression $|u(x)+\bar u|\left(u(x)+V(x)+b{U}_{S(u)}(x) \right)$ in
  the right hand side of \eqref{F} is a bounded,
  $C^{1/2}$--H\"older continuous function, and we can conclude that
  $u\in C^{1,1/2}(\RR^2)$ by \cite[Proposition 2.8]{Silvestre:07}.
  Furthermore, if we assume that $V\in \mathring{H}^{1/2}(\RR^2)\cap
  C^{\alpha}(\RR^2)$ for some $\alpha\in(\frac12,1)$, then repeating a
  similar argument once again we can see that $u\in
  C^{1,\alpha}(\RR^2)$.

  Note, however, that if $u\in C^{1,\alpha}(\RR^2)$, but $u+\bar u$
  changes sign then $|u+\bar u|$, and, hence, the whole right hand
  side of \eqref{F}, is merely a locally Lipschitz function of $x$
  regardless of the smoothness of $V$. Thus, generally speaking, local
  regularity of $u$ can not be improved beyond $C^{1,\alpha}(\RR^2)$.
\end{remark}

\subsection{Proof of Theorem \ref{thm:positive}}

Let $u\in \mathcal H$ be such that
$E(u) = \inf_{\tilde u \in \mathcal H} E(\tilde u)$.  Clearly,
$E(u)\le 0$. In particular,
\begin{equation}\label{E-v-00}
a \|u\|_{\oH^{1/2}(\RR^2)}^2
   + \langle V,{U}_{S(u)}\rangle_{\oH^{1/2}(\RR^2)} +
  \frac{b}{2} \|{U}_{S(u)}\|_{\oH^{1/2}(\RR^2)}^2\le 0.
\end{equation}
Applying Cauchy-Schwarz inequality and then the fractional Sobolev
inequality \eqref{sobolev-inequality}, we conclude that
\begin{equation}\label{V-bound}
  \frac{1}{2b} \|V\|_{\oH^{1/2}(\RR^2)}^2\ge a
  \|u\|_{\oH^{1/2}(\RR^2)}^2\ge a \sqrt{\pi} \, \|u\|_{L^4(\RR^2)}^{2}.
\end{equation}
Similarly, by \eqref{E-v-00} and Cauchy-Schwarz inequality we have
\begin{align}
\label{V-boundbis}
2 \|V\|_{\oH^{1/2}(\RR^2)} \geq b \|{U}_{S(u)}\|_{\oH^{1/2}(\RR^2)}
\geq \pi^{1/4} b \, \|{U}_{S(u)}\|_{L^4(\RR^2)}.
\end{align}

Next assume that the inequality opposite to the one in the
statement of the theorem holds, namely that $\|u\|_{L^\infty(\RR^2)} \ \geq \ \bar u$.
Choose $x^*\in\RR^2$ such that $| u(x^*) | \ge
\frac{1}{2}\|u\|_{L^\infty(\RR^2)}$.  Then $|u+\bar u|\le2\|u\|_\infty$.
Using the same notations as in the proof of Lemma \ref{lem:regularity}, by \eqref{F}, \eqref{V-bound}, \eqref{V-boundbis} and \eqref{sobolev-inequality} we have
\begin{align}
  \| F \|_{L^4(\RR^2)} \leq C \| u \|_{L^\infty(\RR^2)} \| V
  \|_{\mathring{H}^{1/2}(\RR^2)},
\end{align}
for some $C > 0$ depending
only on $a$ and $b$. Therefore by \eqref{e-Holder-plus+} for any $R > 0$ we can write
\begin{equation} {\mathrm{osc}}_{B_R(x^*)} u\le
  C\|u\|_{L^\infty(\RR^2)}\|V\|_{\mathring{H}^{1/2}(\RR^2)}
  R^{1/2},
\end{equation}
again, for some $C > 0$ depending only on $a$ and $b$.

Now, set
\begin{equation}
  R=\frac{c}{\|V\|_{\mathring{H}^{1/2}(\RR^2)}^2},
\end{equation}
where $c>0$ is a constant depending only on $a$ and $b$ chosen in such
a way that $\mathrm{osc}_{B_R(x^*)} u \le
\frac{1}{4}\|u\|_{L^\infty(\RR^2)}$.  Then
\begin{equation}
  \|u\|_{L^4(\RR^2)}^4\ge \int_{B_R(x^*)} u^4 \ud^2 x \ge
  \frac{\pi R^2}{256} \|u\|_{L^\infty(\RR^2)}^4 \ge
  C \|u\|_{L^\infty(\RR^2)}^4 \| V \|_{\mathring{H}^{1/2}(\RR^2)}^{-4},
\end{equation}
for some $C > 0$ depending only on $a$ and $b$, which yields
\begin{equation}
  \|u\|_{L^4(\RR^2)} \| V \|_{\mathring{H}^{1/2}(\RR^2)} \ge C
  \|u\|_{L^\infty(\RR^2)} \geq C \bar u.
\end{equation}
In view of \eqref{V-bound}, we then conclude that
\begin{equation}\label{V-bound2}
\|V\|_{\oH^{1/2}(\RR^2)}\ge C,
\end{equation}
for some $C>0$ depending only on $a$, $b$ and $\bar \rho$, which
completes the proof.  \qed

\section{Proof of Theorems~\ref{thm:Egenplus}, \ref{thm:EA0p} and
  \ref{thm:E0A0}}
\label{sec:th32}

\subsection{Proof of Theorems~\ref{thm:Egenplus} and \ref{thm:EA0p}}

We introduce the function class
\begin{align}
  \mc{H}_+ := \Bigl\{ u \in \mc{H} \, : \, u \geq -\bar{u} \Bigr\},
\end{align}
which is an equivalent way of writing the class
$\mc{A}_{\bar{\rho}}^+$.
To study the variational problem for $E$ on $\mc{H}_+$, let us define
another energy functional $E_+$, given by \eqref{E-v} in which the
functions $\Phi(u)$ and $S(u)$ are replaced by $\Phi_+(u)$ and
$S_+(u)$, respectively. The latter are obtained from the former by a
reflecion around $u = -\bar u$ from the range $u \geq -\bar u$ to $u
\leq -\bar u$ (see Fig. \ref{fig2} and compare with Fig. \ref{fig1}):
\begin{align}
  & S_+(u) := S( \abs{u + \bar{u}} - \bar{u}) =  2 \bar u u + u^2,  \\
  & \Phi_+(u) := \Phi( \abs{u + \bar{u}} - \bar{u}) = \frac{2}{3}
    ( \abs{u + \bar{u}}^3 - \bar{u}^3) - \bar{u} S_+(u).
\end{align}
We also introduce the function class
\begin{align}
  \tilde{\mc{H}} := \Bigl\{ u \in \mathring{H}^{1/2}(\RR^2) :
  S_+(u) \in \mathring{H}^{-1/2}(\RR^2) \},
\end{align}
which is the analog of $\mc{A}_{\bar\rho}$ in the context of $E_+(u)$.
Thus, the energy functional $E_+(u)$, defined for all
$u \in \tilde{\mc{H}}$, is given by
\begin{multline}\label{eq:Eplus}
  E_+(u) := a \norm{u}_{\mathring{H}^{1/2}(\RR^2)}^2
  + \int_{\RR^2} \Phi_+(u(x)) \ud^2 x +
  \average{V, U_{S_+(u)}}_{\mathring{H}^{1/2}(\RR^2)} \\
  + \frac{b}{2} \norm{U_{S_+(u)}}_{\mathring{H}^{1/2}(\RR^2)}^2.
\end{multline}
Note that, by construction, if $u \in \mc{H}_+$ then $u \in
\tilde{\mc{H}}$ and $E(u) = E_+(u)$.

\begin{figure}[t]
  \centering
  \includegraphics[width=2.3in]{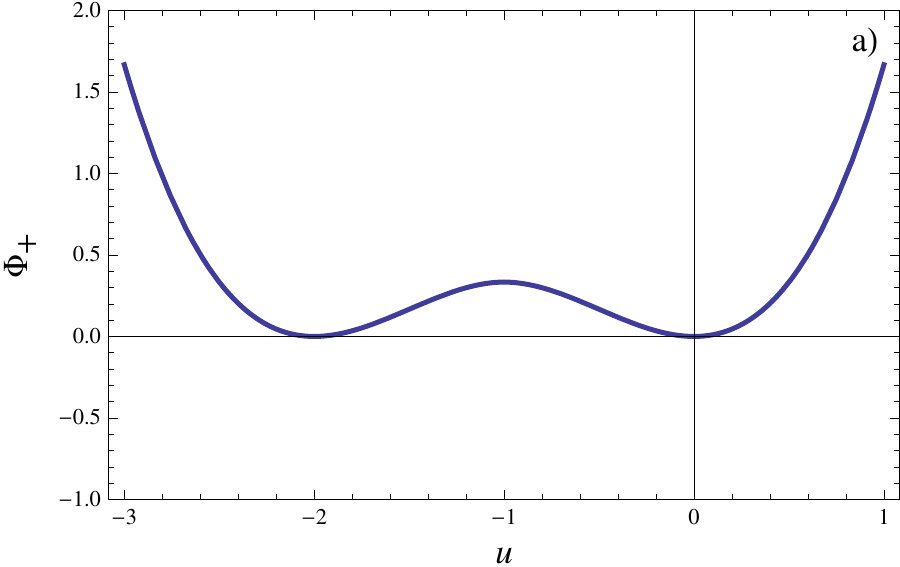}
  \hspace{3mm} \includegraphics[width=2.26in]{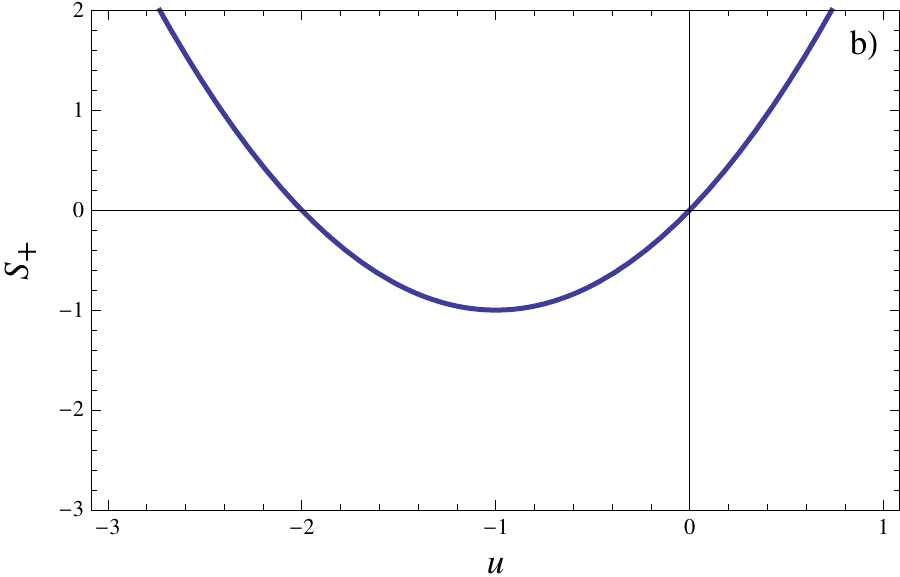}
  \caption{(a) Plot of $\Phi_+(u)$ and (b) plot $S_+(u)$ for $\bar
      u = 1$.}
  \label{fig2}
\end{figure}

Analogous to Proposition~\ref{Frechet}, we have
\begin{proposition}
  If $V \in \mathring{H}^{1/2}(\RR^2)$, then there exists $u_0 \in
  \mc{H}_+$ such that $\displaystyle E_+(u_0) = \inf_{u \in
    \tilde{\mc{H}}} E_+(u)$. Furthermore, $\displaystyle E(u_0) =
  \inf_{u \in \mc{H}_+} E(u)$.
\end{proposition}

\begin{proof}
  Observe first that for any $u \in \tilde{\mc{H}}$, we have
  $\abs{u + \bar{u}} - \bar{u} \in \mc{H}_+ \subset
    \tilde{\mc{H}}$, and by \eqref{Gagliardo}, we have
  \begin{align}
    \bigl\lVert \abs{u + \bar{u}} - & \bar{u}
    \bigr\rVert_{\mathring{H}^{1/2}(\RR^2)}^2 = \frac{1}{4 \pi}
    \iint_{\RR^2 \times \RR^2} \frac{\bigl\lvert \abs{u(x) + \bar{u}}
      - \abs{u(y) + \bar{u}} \bigr\rvert^2}{\abs{x -y}^3} \ud^2 x
    \ud^2
    y \notag \\
    & \hspace{0.5cm} \leq \frac{1}{4 \pi} \iint_{\RR^2 \times \RR^2}
    \frac{\abs{u(x) - u(y)}^2}{\abs{x -y}^3} \ud^2 x \ud^2 y =
    \norm{u}_{\mathring{H}^{1/2}(\RR^2)}^2.
  \end{align}
    Hence,
    \begin{equation}
      E_+( \abs{u + \bar{u}} - \bar{u} ) \leq E_+(u).
    \end{equation}
    Therefore, for a minimizing sequence $\{u_n\}$ of $E_+$ in
    $\tilde{\mc{H}}$, we can consider $\{ \tilde u_n \} :=
    \{\abs{u_n + \bar{u}} - \bar{u} \} \subset \mc{H}_+$, which is
    also a minimizing sequence. The existence of a minimizer then
    follows from the proof of Proposition~\ref{Frechet} by changing
    $S, \Phi$ to $S_+, \Phi_+$ in that proof.

  Finally, $E_+(u) = E(u)$ for $u \in \mc{H}_+$, since $S_+(u)$,
    $\Phi_+(u)$ coincide with $S(u)$, $\Phi(u)$ for $u \geq -\bar{u}$.
  Therefore, the minimizer $u_0$ of $E_+$ (taken to be in $\mc{H}_+$)
  also minimizes $E$ over $\mc{H}_+$. \qed
\end{proof}

It is also clear that any minimizer of $E$ over $\mc{H}_+$ is also a
minimizer of $E_+$ over $\tilde{\mc{H}}$. The advantage of considering $E_+$
is to remove the constraint $u \geq - \bar{u}$ in $\mc{H}_+$. In
particular, we can derive the Euler--Lagrange equation of $E_+$ for a
minimizer $u \in \mc{H}_+$, observing that the arguments in
Section~\ref{sec:EulerLagrange} apply verbatim to the functional $E_+$
(by replacing $S$ and $\Phi$ with $S_+$ and $\Phi_+$,
respectively). If $u \in \mc{H}_+$ is a minimizer of $E_+$, it then
satisfies the Euler-Lagrange equation given in the distributional
sense by
\begin{equation}
  \label{ELupos}
  0 = a (-\Delta)^{1/2} u + \abs{ u + \bar{u} } ( u + V + b U_{S_+(u)} ).
\end{equation}
Note that since $u \geq -\bar{u}$, the absolute value can be
  omitted and $S_+(u)$ coincides with $S(u)$ in the above
equation. Using the Euler--Lagrange equation, we shall establish
additional properties for the minimizer.

\begin{lemma}
  Assume $V\in \mathring{H}^{1/2}(\RR^2)$.  Let
  $u \in \mathcal H_{+}$ be such that $E(u) = \inf_{\tilde u \in
    \mathcal H_+} E(\tilde u)$.  Then $u\in C^{1/2}(\RR^2)\cap
  L^\infty(\RR^2)$, $u(x) \to 0$ as $|x| \to \infty$ and $u(x)>-\bar
  u$ for all $x\in \RR^2$.
\end{lemma}

\begin{proof}
  The regularity follows verbatim from the proof of
  Lemma~\ref{lem:regularity}. Also, since $u \in L^\infty(\RR^2)$, we
  have $S(u) \in L^4(\RR^2)$, and we can again repeat the arguments in
  the proof of Lemma~\ref{lem:regularity}, now applied to
  \eqref{Riesz-repr}, to establish that $U_{S(u)} \in
  C^{1/2}(\RR^2) \cap L^\infty(\RR^2)$ as well.

  Now, since $u$ satisfies \eqref{ELupos} and $u \geq -\bar u$, we
  have that $w := u + \bar{u}\ge 0$ satisfies
  \begin{equation}\label{EU-plus}
    0 = a (-\Delta)^{1/2} w + V w + w ( u + b U_{S(u)}).
  \end{equation}
  Note that by the argument at the beginning of the proof we have
  $w \in L^\infty(\RR^2)$ and
  \begin{equation}
    |u + b U_{S(u)}| \le c,
  \end{equation}
  for some $c>0$ and a.e. $x \in \RR^2$.  Decompose $V=V_+-V_-$, where
  $V_+$ and $V_-$ are the positive and the negative part of $V$,
  respectively. Then
  \begin{equation}
    \label{elM}
    a (-\Delta)^{1/2} w + V_+ w+c w=V_-w+(c-( u + b U_{S(u)} ))w\ge 0.
  \end{equation}
  Since $V_+\in \mathring{H}^{1/2}(\RR^2)\subset L^4(\RR^N)$, the
  potential $V_++c$ belongs to the local Kato class
  $\mathcal K^{1/2}_{loc}$ with respect to $(-\Delta)^{1/2}$, i.e.,
  for every ball $B\subset\RR^2$ we have
  \begin{equation}
    \lim_{\varepsilon\to
      0}\sup_{x\in\RR^2}\int_{B_\varepsilon(x)}\frac{(V_+(y)+c)\chi_{B}(y)}{|x-y|}
    \ud^2 y=0,
  \end{equation}
  see \cite{Kaleta-Lorinczi}*{Definition 2.1} or
  \cite{Carmona-Simon}*{Section III and Theorem III.1(iii)}.  Then
  standard methods of semigroup theory (see
  e.g. \cite{Reed-Simon-4}*{Section XIII.12}) can be used to show that
  $a (-\Delta)^{1/2} + V_+ +c $ defines a self-adjoint
  positive-definite linear operator in $L^2(\RR^2)$
  \cite{Kaleta-Lorinczi}*{Theorem 2.1}. Moreover, the Green's function
  $G_{a, c}^{V_+}(x,y):\RR^2\times\RR^2\to\RR$ of
  $a (-\Delta)^{1/2} + V_+ +c $ is well defined and strictly positive
  (cf. \cite{Kaleta-Lorinczi}*{Lemma 2.1(4) and Section 2.4}). In
  addition, since $V_+\ge 0$, the Green's function
  $G_{a, c}^{V_+}(x,y)$ is dominated by the Green's function
  $G_{a, c}(|x-y|)$ of the operator $a (-\Delta)^{1/2} +c$, so that we
  have
  \begin{equation}\label{Green-V}
    0<G_{a, c}^{V_+}(x,y)\le G_{a, c}(|x-y|)\quad\text{for all
      $x,y\in\RR^2$}.
  \end{equation}
  The function $G_{a, c}(r)$ is given explicitly by
  \begin{align}
    \label{Gac}
    G_{a, c}(r) := \frac{c}{4 a^2} \left( \frac{2 a}{\pi c r}
    -\pmb{H}_0\left(\frac{c r}{a}\right) + Y_0\left(\frac{c
    r}{a}\right) \right),
  \end{align}
  where $\pmb{H}_0(z)$ is the Struve function, $Y_0(z)$ is the Bessel
  function of the second kind, which can be obtained using Fourier
  transform.  Moreover, $G_{a, c}$ obeys \cite[Theorem 3.3 and Lemma
  4.1]{FelmerQuaasTan:12}
  \begin{equation}\label{Tan}
    G_{a, c}(r) \sim
    \begin{cases}
      \ r^{-1}, & r \ll 1, \\
      \ r^{-3}, & r \gg 1,
    \end{cases}
  \end{equation}
  and, therefore, we have
  $G_{a,c}^{V_+} \in L^{4/3}(\RR^2) \cap L^1(\mathbb R^2)$.  Denoting
  the right-hand side of \eqref{elM} by $g(x)\ge 0$, since
  $g \in L^4(\RR^2) + L^\infty(\RR^2)$ we then have distributionally
  and a.e. in $\RR^2$
  \begin{equation}
    \label{wGaM}
    w(x) = \int_{\RR^2}G_{a, c}^{V_+}(x,y) g(y) \ud^2 y.
  \end{equation}
  This implies that $w$ is strictly positive.  \qed
\end{proof}

\begin{remark}\label{remark-regularity-plus}
  Note that unlike minimizers in $\mathcal H$ (see Remark
  \ref{remark-regularity}), further regularity of minimizers
  $u\in\mathcal H_+$ is expected under additional smoothness
  hypothesis on $V$.  For example, if we assume that $V\in
  \mathring{H}^{1/2}(\RR^2)\cap C^\infty(\RR^2)$ then $u\in
  C^\infty(\RR^2)$. Indeed, if $u\in C^{k,\alpha}(\RR^2)$ for some
  $k\ge 1$ and $w:=u+\bar u>0$ then $w\left(u+V+b{U}_{S(u)}\right)\in
  C^{k,\alpha}(\RR^2)$.  Thus, differentiating the Euler--Lagrange
  equation in \eqref{EU-plus} with respect to $x$ and applying
  \cite[Proposition 2.8]{Silvestre:07}, we conclude that $u\in
  C^{k+1,\alpha}(\RR^2)$.  This argument can be iterated infinitely
  many times, we omit the details.
\end{remark}

Finally, we show that the minimizer of $E$ on $\mc{H}_+$ is unique. It
is more convenient to rewrite the energy functional using $\rho$ as
the variable, as in the uniqueness proof for the usual
Thomas-Fermi-von~Weizs\"acker model \cite{Lieb:81}. We write
\begin{multline}\label{eq:energyrho}
  E(\rho) = a \norm{ \sqrt{\rho} - \sqrt{\bar{\rho}}
  }_{\mathring{H}^{1/2}(\RR^2)}^2
  + \int_{\RR^2} \Phi \bigl( \sqrt{\rho(x)} - \sqrt{\bar{\rho}} \bigr) \ud^2 x \\
  + \average{V, U_{S(\sqrt{\rho} -
      \sqrt{\bar{\rho}})}}_{\mathring{H}^{1/2}(\RR^2)} + \frac{b}{2}
  \norm{U_{S(\sqrt{\rho} -
      \sqrt{\bar{\rho}})}}^2_{\mathring{H}^{1/2}(\RR^2)}.
\end{multline}
Note that
\begin{equation}
  S(\sqrt{\rho} - \sqrt{\bar{\rho}}) = \rho - \bar{\rho},
\end{equation}
and is, therefore, linear in $\rho$, and
\begin{equation}
  \Phi(\sqrt{\rho} - \sqrt{\bar{\rho}}) = \frac{2}{3} (\rho^{3/2} -
  \bar{\rho}^{3/2}) - \sqrt{\bar{\rho}}( \rho - \bar{\rho})
\end{equation}
is strictly convex in $\rho$. Hence, the last three terms in $E(\rho)$
given by \eqref{eq:energyrho} are convex on
$\mc{A}_{\bar{\rho}}^+$. Moreover, even though $\sqrt{\rho}$ is a
concave function of $\rho$, the following lemma shows that
$\norm{\sqrt{\rho} - \sqrt{\bar{\rho}}}_{\mathring{H}^{1/2}(\RR^2)}^2$
is convex, and the energy $E(\rho)$ is strictly convex. The uniqueness
of the minimizer then follows.

\begin{lemma}
  The set $\mc{A}_{\bar{\rho}}^+$ is convex. Furthermore, the
  functional $E(\rho)$ defined in \eqref{eq:energyrho} is strictly
  convex on $\mc{A}_{\bar{\rho}}^+$, i.e., for every $\rho_0, \rho_1
  \in \mc{A}_{\bar{\rho}}^+$, $\rho_0 \not= \rho_1$, and every $t
  \in (0, 1)$, there holds
  \begin{equation}
    E( t \rho_0 + (1-t) \rho_1) \ < \  t E(\rho_0) + (1-t) E(\rho_1).
  \end{equation}
\end{lemma}

\begin{proof}
Denote $\rho_t = t \rho_0 + (1-t) \rho_1$.
From \cite{lieb96}*{Theorem 7.13} it follows that $\sqrt{\rho_t} - \sqrt{\bar \rho} \in
  \mathring{H}^{1/2}(\RR^2)$ and
  \begin{equation}\label{eq:convexH1/2}
    \norm{ \sqrt{\rho_t} - \sqrt{\bar{\rho}}}_{\mathring{H}^{1/2}(\RR^2)}^2 \leq t
    \norm{ \sqrt{\rho_0} - \sqrt{\bar{\rho}}}_{\mathring{H}^{1/2}(\RR^2)}^2 + (1-t)
    \norm{ \sqrt{\rho_1} - \sqrt{\bar{\rho}}}_{\mathring{H}^{1/2}(\RR^2)}^2.
  \end{equation}
Also, clearly $\rho_t - \bar{\rho} \in
 \mathring{H}^{-1/2}(\RR^2)$ and $\rho_t \geq 0$. Hence, $\rho_t \in
  \mathcal A_{\bar \rho}^+$, implying that $\mc{A}_{\bar{\rho}}^+$ is
  a convex set. The strict convexity of $E(\rho)$ then follows from
  the strictly convexity of $\Phi$ in the second term in
  $E(\rho)$. \qed
\end{proof}

\subsection{Proof of Theorem~\ref{thm:E0A0}}

For $u \in \mc{H}_+$, we have $E(u) = E_+(u)$, where $E_+$ is
  defined in \eqref{eq:Eplus} with the specific choice $\bar u=0$:
\begin{multline}\label{eq:Eplus0}
  E_+(u) = a \norm{u}_{\mathring{H}^{1/2}(\RR^2)}^2
  + \frac{2}{3}\int_{\RR^2} \abs{u(x)}^3 \ud^2 x \\
  - \int_{\RR^2} \frac{\abs{u(x)}^2}{\bigl(1 + \abs{x}^2\bigr)^{1/2}}
  \ud^2 x + \frac{b}{2}
  \norm{U_{\abs{u}^2}}_{\mathring{H}^{1/2}(\RR^2)}^2.
\end{multline}
In view of Theorem \ref{thm:EA0p}, in order to prove Theorem
\ref{thm:E0A0} it is sufficient to show that:

$(i)$ If $a \geq a_c$, then $E(u) > 0$ for every non-zero $u \in
\mc{H}$,

$(ii)$ If $a < a_c$, then $\inf_{u \in \mc{H}_+} E_+(u)
<0$.

\noindent
Claim $(i)$ follows directly from the fractional Hardy's inequality
\begin{equation}\label{eq:Hardy-Frank}
  a_c\norm{u}_{\mathring{H}^{1/2}(\RR^2)}^2\ge \int_{\RR^2}
  \frac{\abs{u(x)}^2}{ \abs{x} } \ud^2 x,
\end{equation}
which is valid for all $u\in \mathring{H}^{1/2}(\RR^2)$ with the
optimal constant $a_c= \frac{\Gamma^2(1 / 4)}{2 \Gamma^2(3 / 4)}$, see
\cite[Remark 4.2]{frank08}.

Claim $(ii)$ is a consequence of the following.

\begin{lemma}\label{lemma:Hardy}
  Let $c<a_c$. Then there exists $u_c\in C^\infty_c(\RR^2)$
  such that $u_c\ge 0$ and
  \begin{equation}\label{eq:E-Hardy-sharp}
    c\norm{u_c}_{\mathring{H}^{1/2}(\RR^2)}^2< \int_{\RR^2}
    \frac{\abs{u_c(x)}^2}{\bigl(1 + \abs{x}^2\bigr)^{1/2}} \ud^2 x.
  \end{equation}
\end{lemma}

Indeed, let $a<a_c$. Then, using Lemma \ref{lemma:Hardy} with some
$c\in(a,a_c)$, for all sufficiently small $t>0$ we obtain
\begin{multline}\label{eq:Eplus0-1}
  E_+(tu_c) < -(c-a)t^2 \norm{u_c}_{\mathring{H}^{1/2}(\RR^2)}^2\\
  + \frac{2t^3}{3}\int_{\RR^2}\abs{u_c(x)}^3 \ud^2 x + \frac{bt^4}{2}
  \norm{U_{\abs{u_c}^2}}_{\mathring{H}^{1/2}(\RR^2)}^2<0.
\end{multline}
We conclude that $\inf_{u \in \mc{H}_+} E_+(u) <0$, which
proves Claim $(ii)$.

\medskip

We are only left to prove Lemma \ref{lemma:Hardy}.

\begin{proof}[of Lemma \ref{lemma:Hardy}]
Let $u\in C^\infty_c(\RR^2)$ be such that
\begin{equation}\label{eq:Hardy2}
  c\norm{u}_{\mathring{H}^{1/2}(\RR^2)}^2-\int_{\RR^2}
  \frac{\abs{u(x)}^2}{\abs{x}} \ud^2 x\le-1
\end{equation}
(cf. \cite[Remark 4.2]{frank08}, where one can choose $u\in
C^\infty_c(\RR^2)$ as a suitable approximation of $|x|^{-1/2}$).
For $\lambda>0$, set $u_\lambda(x)=u(x/\lambda)$.
Then
\begin{multline}\label{eq:Hardy3}
  c\norm{u_\lambda}_{\mathring{H}^{1/2}(\RR^2)}^2-\int_{\RR^2}
  \frac{\abs{u_\lambda(x)}^2}{\bigl(1 + \abs{x}^2\bigr)^{1/2}}\ud^2
  x=\\
  \lambda \left( c\norm{u}_{\mathring{H}^{1/2}(\RR^2)}^2- \int_{\RR^2}
    \frac{\abs{u(y)}^2}{\bigl(\lambda^{-2} +
      \abs{y}^2\bigr)^{1/2}}\ud^2 y \right) \le-{\lambda \over 2},
\end{multline}
for all sufficiently large $\lambda>0$, in view of \eqref{eq:Hardy2}
and the monotonicity of the mapping
$\lambda\mapsto ( \lambda^{-2} + \abs{y}^2)^{-1/2}$. \qed
\end{proof}

\section*{Acknowledgments}

\noindent The authors wish to thank an anonymous referee for helpful
suggestions. JL would like to acknowledge support from the Alfred
P.~Sloan Foundation and the National Science Foundation under award
DMS-1312659. CBM was supported, in part, by the National Science
Foundation via grants DMS-0908279 and DMS-1313687.

\bibliographystyle{amsxport}
\bibliography{graphene}

\end{document}